\documentclass[11pt,a4paper]{article}
\usepackage[DIV=13]{typearea}
\usepackage[onehalfspacing]{setspace}
\usepackage[english]{babel}
\usepackage[ansinew]{inputenc}
\usepackage{amssymb,amsmath,bm,array,dsfont,graphicx,natbib}
\usepackage[colorlinks,citecolor=black,linkcolor=black,urlcolor=black]{hyperref}
\usepackage[hang,tight]{subfigure}
\usepackage{authblk}
\usepackage{framed}
\usepackage{xcolor}
\usepackage{bbm}
\usepackage{longtable}
\usepackage{rotating}
\usepackage{amsthm}

\colorlet{shadecolor}{gray!25}

\newcommand{\newoperator}[3]{\newcommand*{#1}{\mathop{#2}#3}}
\newcommand{\renewoperator}[3]{\renewcommand*{#1}{\mathop{#2}#3}}
\newcommand{\mA}{A}

\newcommand{\mD}{D}

\newcommand{\mH}{H}

\newcommand{\vv}{v}

\newcommand{\vtheta}{\theta}

\newcommand{\mTheta}{\varTheta}

\newcommand{\mSigma}{\varSigma}

\newcommand{\mOmega}{\varOmega}
\renewoperator{\Re}{\mathrm{Re}}{\nolimits}
\renewoperator{\Im}{\mathrm{Im}}{\nolimits}
\makeatletter
\newcommand{\rd}{\@ifnextchar^{\DIfF}{\DIfF^{}}}
\def\DIfF^#1{%
   \mathop{\mathrm{\mathstrut d}}%
   \nolimits^{#1}\gobblespace}
\def\gobblespace{\futurelet\diffarg\opspace}
\def\opspace{%
   \let\DiffSpace\!%
   \ifx\diffarg(%
   \let\DiffSpace\relax
   \else
   \ifx\diffarg[%
   \let\DiffSpace\relax
   \else
   \ifx\diffarg\{%
   \let\DiffSpace\relax
   \fi\fi\fi\DiffSpace}

\newcommand{\E}{\operatorname{E}}

\newcommand{\Var}{\operatorname{Var}}
\newcommand{\Corr}{\operatorname{Corr}}
\newcommand{\Cov}{\operatorname{Cov}}

\newoperator{\ip}{\mathrm{int}}{\nolimits}

\newcommand{\plim}{\operatorname{plim}}

\newcommand{\pto}{\stackrel{p}{\longrightarrow}}
\newcommand{\dto}{\stackrel{d}{\longrightarrow}}

\newtheorem{theorem}{Theorem}[section]
\newtheorem{assumption}{Assumption}
\newtheorem{lemma}[theorem]{Lemma}

\newcommand{\beq}{\begin{equation}}
\newcommand{\eeq}{\end{equation}}
\newcommand{\bal}{\begin{align*}}
\newcommand{\eal}{\end{align*}}

\newcommand{\bvec}{\begin{pmatrix}}
\newcommand{\evec}{\end{pmatrix}}
\newcommand{\bmat}{\begin{bmatrix}}
\newcommand{\emat}{\end{bmatrix}}

\newcommand{\bsmat}{\begin{smallmatrix}}
\newcommand{\esmat}{\end{smallmatrix}}

\usepackage[bottom]{footmisc}
\usepackage{enumitem}

\usepackage{titlesec}

\allowdisplaybreaks
\title{Monitoring the pandemic: A fractional filter for the COVID-19 contact rate}
\author[a,b]{Tobias Hartl\footnote{Corresponding author: Department of Economics and Econometrics, University of Regensburg, Universit\"atsstr. 31, 93053 Regensburg, Germany, 
		email: tobias1.hartl@ur.de}}
\affil[a]{University of Regensburg, 93053 Regensburg, Germany}
\affil[b]{Institute for Employment Research, 90478 Nuremberg, Germany}
\date{February 2021}

\begin{document}
\maketitle

\thispagestyle{empty}
\setcounter{page}{0}
\vspace{-1cm}
\paragraph{\bf Abstract.}
\begin{spacing}{1.15}
	This paper aims to provide reliable estimates for the COVID-19 contact rate of a Susceptible-Infected-Recovered (SIR) model. From observable data on confirmed, recovered, and deceased cases, a noisy measurement for the contact rate can be constructed.
	To filter out measurement errors and seasonality, a novel unobserved components (UC) model is set up. It specifies the log contact rate as a latent, fractionally integrated process of unknown integration order. The fractional specification reflects key characteristics of aggregate social behavior such as strong persistence and gradual adjustments to new information. A computationally simple modification of the Kalman filter is introduced and is termed the fractional filter. It allows to estimate UC models with richer long-run dynamics, and provides a closed-form expression for the prediction error of UC models. Based on the latter, a conditional-sum-of-squares (CSS) estimator for the model parameters is set up that is shown to be consistent and asymptotically normally distributed. The resulting contact rate estimates for several countries are well in line with the chronology of the pandemic, and allow to identify different contact regimes generated by policy interventions.
	As the fractional filter is shown to provide precise contact rate estimates at the end of the sample, it bears great potential for monitoring the pandemic in real time. 
\end{spacing}

\vspace{-0.1cm}
\paragraph{\bf Keywords.}
COVID-19, filtering, long memory, SIR model, unobserved components.

\vspace{-0.1cm}
\paragraph{\bf JEL-Classification.}
C22, C51, C52.

\newpage

\section{Introduction}
Since the outbreak of COVID-19 reducing social contacts is widely viewed as the key way to contain the spread of the virus.
In terms of the Susceptible-Infected-Recovered (SIR) model\footnote{The SIR model -- in its various variants -- has recently become a popular tool to study the economic impact of the pandemic and for policy simulations, see \citet{AceCheWe2020, AveBosCl2020, Kor2021, LiuMooSc2021} among others.}, this relates to the contact rate, defined as the average number of contacts per person per time unit multiplied by the probability of disease transmission between a susceptible and an infected individual \citep{Het2000}. The probability of disease transmission should only depend on characteristics that are specific to the virus. Therefore, the contact rate can be interpreted as a proxy for aggregate social behavior and is the key variable addressed by social distancing measures. Knowing the trajectory of the contact rate would allow to draw inference on the impact of policy measures on contact reduction, to real-time monitor the dynamics of virus dispersion, and to design policy rules based on the current pandemic situation. Since the contact rate itself is unobservable, appropriate methods to estimate the contact rate are required, and will be considered in this paper.

At the early stage of the pandemic, first estimates for the natural logarithm of the contact rate were obtained by fitting a deterministic, linear time trend with structural breaks to transformations of data on confirmed, recovered, and deceased cases \citep{HarWaeWe2020, LeeLiaSe2021, LiuMooSc2021}. Modeling the log contact rate by a piece-wise linear time trend was a reasonable and pragmatic approximation given the short time series on case numbers available at that time. However, it implies that contact rate growth evolves deterministically as a straight line with jumps at the break dates. This assumption is likely to be violated by the behavior of individuals. While structural breaks may be suitable to identify turning points of the contact rate, they are inappropriate for monitoring the current pandemic situation, as breaks require at least some post-break observations to be well identified.

This paper aims to improve estimates for the contact rate of COVID-19 by taking into account key features of aggregate social behavior. In detail, the log contact rate, as denoted by $\log \beta_t$, is modeled as an unobserved, fractionally integrated process of (unknown) order $d \in \mathbb{R}_+$, generated by stochastic shocks $\{\eta_i\}_{i=1}^t$.\footnote{Fractional integration techniques have been found useful for describing the aggregate behavior of individuals in a variety of applications, e.g.\ for explaining the Deaton paradox \citep{DieRud1991} and for the estimation of the business cycle \citep{HarTscWe2020}.} 
The stochastic specification of the contact rate is motivated by the consideration that social decisions, e.g.\ on whether to meet, are made conditional on the information available at that time, e.g.\ on current social distancing measures or the state of the pandemic. As information does not evolve deterministically but appears as stochastic shocks, this suggests to treat $\log \beta_t$ as a stochastic process generated by the information shocks $\{\eta_i\}_{i=1}^t$. 
Specifying $\log \beta_t$ as a fractionally integrated process accounts for strong persistence and nonstationarity (in short:\ long memory) of social behavior. In contrast to structural breaks but also to random walks, the fractional specification allows social behavior to gradually adjust to new information both at the individual and at the aggregate level. Individually, this reflects a gradual reduction or increase of contacts as new information becomes available (e.g.\ as new contact restrictions are imposed), while on aggregate it allows individuals to react heterogeneously both in terms of speed and intensity to novel information. As the persistence of the log contact rate is unknown, the integration order $d$ is treated as an unknown parameter to be estimated. 

Methodologically, this paper contributes to the literature on time series filtering by setting up a novel unobserved components (UC) model that does not require prior knowledge about the integration order of the variable under study. Current UC models and related filtering techniques rely heavily on prior assumptions about the integration order $d$ and typically assume $d=1$ \citep[e.g.][]{Har1985, MorNelZi2003, ChaMilPa2009} or $d=2$ \citep[e.g.][]{Cla1987, HodPre1997, OhZivCr2008} to be known. In contrast, the novel UC model reflects that the degree of persistence of the log contact rate is unknown. It allows to decompose a noisy measurement for the log contact rate that is based on a transformation of data on confirmed, recovered, and deceased cases, into measurement errors, seasonal components, and the unobserved log contact rate itself. As the latter is modeled by a fractionally integrated process, the model is called the fractional UC model. 

The second methodological contribution of this paper is to derive a computationally much simpler estimator for the model parameters and the unobserved components compared to current state space methods. Current methods typically rely on the Kalman filter to set up a conditional \mbox{(quasi-)likelihood} function for the estimation of the model parameters. Given the parameter estimates, a time-varying signal for the unobserved components is then obtained from the Kalman smoother. Both the Kalman filter and smoother become computationally infeasible when the dimension of the state vector of UC models is high, as for fractionally integrated processes. To address this problem, this paper proposes a computationally simple modification of the Kalman filter and smoother that is termed the fractional filter. While filtered and smoothed estimates from the fractional filter are identical to the Kalman filter and smoother, the fractional filter avoids the computationally intensive recursions for the conditional variance. The fractional filter provides a closed-form expression for the prediction error of UC models, based on which a 
conditional-sum-of-squares (CSS) estimator for the fractional integration order and other model parameters is set up. While the CSS estimator has been found useful for the estimation of ARFIMA models, see \cite{HuaRob2011} and \cite{Nie2015}, it has not been considered in the UC literature so far. The CSS estimator minimizes the sum of squared prediction errors that is proportional to the exponent in the conditional \mbox{(quasi-)likelihood} function based on the Kalman filter. Due to the computational gains from the fractional filter, the CSS estimator allows to estimate UC models with richer long-run dynamics. 
The paper provides the asymptotic theory for the CSS estimator, showing it to be consistent and asymptotically normally distributed, while the finite sample properties are assessed by a Monte Carlo study.

Using data from the Johns Hopkins University Center for Systems Science and Engineering \citep[JHU CSSE]{DonDuGa2020}, estimates for contact and reproduction rate are presented for Canada, Germany, Italy, and the United States, where benefits from the new methods directly become apparent: 
First, estimation results are not only well in line with the chronology of the pandemic, but also allow to identify different contact regimes generated by the strengthening and easing of contact restrictions. Second, a recursive window evaluation shows contact rate estimates at the end of a truncated sample to largely overlap with those based on the full sample information. This makes the fractional filter a suitable candidate for monitoring outbreaks at the current frontier of the data. And third, the proposed estimation and filtering techniques are shown to be fairly robust to under-reporting of recovered cases, which is of particular importance for the US, as several states do not report data on recovered individuals. While under-reporting heavily downward-biases contact and reproduction rate estimates in \cite{LeeLiaSe2021}, this is shown not to be the case for the fractional filter.

The remaining paper is organized as follows: Section \ref{model} motivates the specification of the contact rate and sets up the fractional UC model. Section \ref{Ch:est} introduces the fractional filter for $\log \beta_t$, covers parameter estimation via the CSS estimator and presents the asymptotic theory. Section \ref{application} contains empirical results for Canada, Germany, Italy, and the United States, while section \ref{Conclusion} concludes. The appendices include proofs for consistency and asymptotic normality of the CSS estimator as well as a Monte Carlo study on the finite sample properties. 
\section{A fractional unobserved components model for the contact rate}\label{model}
To motivate the estimation of the contact rate, consider the discrete SIR model, augmented to include deaths, which also forms the starting point of \citet[eqn.\ 1--4]{Pin2020} and \citet[eqn.\ 2.1]{LeeLiaSe2021}
\begin{align}
	1 &= S_t + I_t + D_t + R_t, 		\label{eqn:1}\\
	\Delta I_t &= \beta_t S_{t-1} I_{t-1} - \gamma I_{t-1}, 		\label{eqn:2}\\
	\Delta D_t &= \gamma_d I_{t-1}, 		\label{eqn:3}\\
	\Delta R_t &= \gamma_r I_{t-1}. 		\label{eqn:4}
\end{align}
In \eqref{eqn:1}, the (initial) population size, normalized to be one, is decomposed into $S_t$, the proportion of the population susceptible in $t$, $I_t$, the fraction of the population infected in $t$, $D_t$, the fraction that has died until $t$, and $R_t$, the proportion that has recovered until $t$. In \eqref{eqn:2}, $\gamma = \gamma_d + \gamma_r$ denotes the rate at which infected either die, see \eqref{eqn:3}, or recover, see \eqref{eqn:4}, and obviously $\gamma_d, \gamma_r \geq 0$. Thus, $\gamma I_{t-1}$ denotes the fraction of outflows of infected at $t$. The fraction of new infections at $t$ is captured by $\beta_t S_{t-1}I_{t-1}$, where $S_{t-1}I_{t-1}$ can be interpreted as the average probability of a contact being between a susceptible subject and an infected subject.
$\beta_t > 0$ is called the contact (or transmission) rate. It equals the average number of contacts per person per time unit multiplied by the probability of disease transmission between a susceptible and an infectious person \citep{Het2000}. 
As in \cite{LeeLiaSe2021}, the contact rate is allowed to be time-varying. This reflects social behavior to change over time, e.g.\ in response to policy changes or to novel information on the pandemic. Since the contact rate determines inflows into infected, see \eqref{eqn:2}, it is the key variable tackled by social distancing policies.

Based on the contact rate, the reproduction rate $\mathcal{R}_{ t} = \beta_t/\gamma$ can be derived. It is the average number of infections caused by an infected subject during the infectious period $1/\gamma$ at the early stage of the pandemic (where $S_{t-1} \approx 1$). $\mathcal{R}_{t}$ is an indicator for the current dynamics of the pandemic, as for $\mathcal{R}_{t} < 1$ outflows from infected exceed inflows, causing $\Delta I_t$ to converge, see \eqref{eqn:2} where $0 \leq S_{t-1}  \leq 1$. Thus, if policy seeks to contain the spread of COVID-19, then it must control the contact rate, which controls the reproduction rate $\mathcal{R}_t$.

As shown by \cite{LeeLiaSe2021}, from \eqref{eqn:1} to \eqref{eqn:4} a measurement for the contact rate $\beta_t$ can be obtained directly: Denote $C_t = I_t +  R_t + D_t$ as the fraction of confirmed cases (consisting of infected, recovered, and deceased cases) and use $\Delta C_t = \Delta I_t  + \Delta R_t + \Delta D_t$ together with \eqref{eqn:2} to \eqref{eqn:4} to obtain $\Delta C_t = \beta_t S_{t-1} I_{t-1} - \gamma I_{t-1} + (\gamma_d + \gamma_r) I_{t-1} =  \beta_t S_{t-1} I_{t-1} $. Solving for $\beta_t$ yields 
\begin{align}\label{eqn:5}
	\beta_t = \frac{\Delta C_t}{I_{t-1}S_{t-1}} =: Y_t,
\end{align}
see \citet[eqn.\ 2.2]{LeeLiaSe2021}. As argued there, if for each $t$ the data $(C_t, R_t, D_t)$ can be observed, then the time-varying contact rate can be calculated straightforwardly via \eqref{eqn:5} using $S_t = 1 - C_t$, as well as $I_t = C_t - R_t - D_t$. 

Unfortunately, reported case numbers for $C_t$, $R_t$, and $D_t$, such as the daily data from JHU CSSE used in the applications in section \ref{application}, suffer from measurement errors, see e.g.\ \citet{HorLiuSc2021}. In addition, they display a strong weekly seasonal pattern that is likely to be driven by a varying number of tests conducted over the different days of the week \citep{BerSelAg2020}. 
Under the assumption that $Y_t$ is measured with a proportionally constant error variance resulting from seasonality and measurement errors, one has the following structure for the natural logarithm of the observable $\tilde Y_t$.
\begin{assumption}[Multiplicative seasonal and measurement errors] \label{asu:1}
	For each $t$, the observable $\tilde Y_t$ satisfies
	\begin{align*}
	\log	\tilde Y_t = \log Y_t + \sum_{i=1}^7 \alpha_i s_{i, t} + u_t = \log \beta_t + \sum_{i=1}^7 \alpha_i s_{i, t} + u_t, \qquad t=1,...,n,
	\end{align*}
	with $Y_t$ as given in \eqref{eqn:5}. $s_{i,t}$ are seasonal dummies for $i=1,...,7$, that capture the weekly patterns of reported case numbers, $\sum_{i=1}^7 \alpha_i = 0$, and the measurement error $u_t\sim WN(0, \sigma_u^2)$ is white noise.
\end{assumption}
Assumption \ref{asu:1} specifies an unobserved components (UC) model where the observable noisy measurement $\log \tilde Y_t$ is decomposed into an unobservable measurement error $u_t$, seasonal components $\sum_{i=1}^7 \alpha_i s_{i, t}$, and the log contact rate $\log \beta_t$. The log specification accounts for a proportional impact of measurement errors and seasonality, and forces the contact rate to be strictly positive. 

As the different components are not separately identified, an additional assumption on the dynamic structure of the contact rate is required. Empirical models of COVID-19 case numbers have so far assumed $\log \beta_t$ to follow a piece-wise linear time trend with structural breaks, see \citet[][]{HarWaeWe2020, LeeLiaSe2021, LiuMooSc2021}. As an alternative, the UC literature suggests to model time-varying coefficients as random walks \citep[see][for an overview]{DurKoo2012}. Both specifications assume contact rate growth $\Delta \log \beta_t$ only to be contemporaneously affected either by structural breaks or by stochastic shocks, an assumption that is likely to be violated. 
Reflecting that the persistence properties of social behavior, and thus of the contact rate, are unknown, assumption \ref{asu:2} specifies the log contact rate as a fractionally integrated process of unknown order $d$.

\begin{assumption}[Specification of the contact rate] \label{asu:2}
	The log contact rate follows a type II fractionally integrated process of order $d\in \mathbb{R_+}$,  denoted as $\log \beta_t \sim I(d)$, where 
	\begin{align*}
	 \log \beta_t = \mu + x_t, \qquad \Delta_+^d x_t = \eta_t, \qquad \eta_t \sim WN(0, \sigma_\eta^2), \qquad t=1,...,n,
	\end{align*}
	$\mu$ is an intercept, and the $\eta_t$ are white noise and are independent of the measurement error $u_t$. 
\end{assumption}
Under assumption \ref{asu:2}, the log contact rate $\log \beta_t$ is a stochastic long memory process generated by the shocks $\{\eta_i\}_{i=1}^t$. The shock $\eta_t$ models the information new in $t$, such as news reports or policy announcements. Social decisions, reflected in $\log \beta_t$, however may additionally depend on past information $\eta_{t-1},...,\eta_1$. Together, $\{\eta_i\}_{i=1}^t$ forms the information available at $t$, conditional on which social decisions, e.g.\ on whether to meet, are made. The specification takes into account that new information does not evolve deterministically, but appears as stochastic shocks, which cannot be captured by a deterministic specification as e.g.\ in \cite{LeeLiaSe2021}.

The degree of persistence of the log contact rate is determined by the integration order $d$, which controls for the persistent impact of past shocks via the fractional difference operator $\Delta^d_+$. 
The latter exhibits a polynomial expansion in the lag operator $L$ of order infinite
\begin{align}\label{fracdiff}
	\Delta^{d} &= (1-L)^{d} = \sum_{i = 0}^{\infty}\pi_{i}(d)L^{i},  \qquad
	\pi_{i}(d) = 
	\begin{cases}
		\frac{i-d-1}{i}\pi_{i-1}(d) &  i = 1, 2, ..., \\ 
		1										&   i = 0.
	\end{cases} 
\end{align}
The $+$-subscript denotes a truncation of an operator at $t \leq 0$, $\Delta_+^d x_t= \sum_{i=0}^{t-1}\pi_i(d) x_{t-i}$, which reflects the type II definition of fractionally integrated processes \citep{MarRob1999}.
For $d=1$ the log contact rate is a random walk, which follows from plugging $d=1$ into \eqref{fracdiff}. Consequently, assumption \ref{asu:2} encompasses the predominant specification in the UC literature. However, assumption \ref{asu:2} allows for a far more general dynamic impact of past shocks $\eta_1,...,\eta_t$ on $\log \beta_t$, as can be seen by plugging $x_t = \Delta_+^{-d}\eta_t$ into $\log \beta_t = \mu + x_t$, which gives 
\begin{align}\label{eqn:6}
	\log \beta_t = \mu + \Delta_+^{-d}\eta_t = \mu + \sum_{i=0}^{t-1} \pi_i(-d)\eta_{t-i}.
\end{align}
While a random walk is an unweighted sum of past shocks $\eta_1,...,\eta_t$, so that $\pi_i(-1)=1$ for all $i=1,...,t-1$, allowing $d \neq 1$ yields non-uniform weights of past shocks in the impulse response function of $\log \beta_t$ and thus a gradual adjustment of the log contact rate to new information. This reflects that social behavior adjusts step-wise to new information both at the individual and the aggregate level. As processing new information on the Coronavirus and revising individual decisions (e.g.\ meeting friends, traveling, working from home) takes time and evolves gradually, individuals can be expected to step-wise adjust their contacts in response to new information. Overall, individuals will react heterogeneously both in terms of speed and intensity to novel information: Some will anticipate new information faster than others, and the extent of reaction will depend on individual characteristics such as risk awareness and attitudes. Such gradual adjustments are well captured by assumption \ref{asu:2}, in particular when $1 < d < 2$: 
In that case, contact rate growth $\Delta \log \beta_t \sim I(d-1)$ is strongly persistent and mean-reverting, as will become apparent in the applications in section \ref{application}. Strong persistence reflects the gradual adjustment of social behavior to new information, while mean-reversion ensures an asymptotically declining impact of past information to today's contact rate growth.


The remaining assumptions are imposed mainly for technical reasons. The type II definition of fractional integration  assumes zero starting values for the fractionally integrated process by truncating the polynomial expansion of the fractional difference operator, $\Delta_+^d x_t= \sum_{i=0}^{t-1}\pi_i(d) x_{t-i}$. It is required to treat the asymptotically stationary ($d<1/2$, from now on `stationary' for brevity) and the asymptotically nonstationary case ($d > 1/2$, from now on `nonstationary') alongside each other. While the type II definition may be a strong assumption for some time series, it is plausible for the contact rate, as we have data covering roughly the whole pandemic. Thus, the pre-sample shocks $\eta_i$, $i \leq 0$, should be zero. Independence of $u_t$ and $\eta_t$ follows from the characterization of $u_t$ as a measurement error that should not influence the contact rate. In general, the assumption can be relaxed to allow for  $\Corr(\eta_t, u_t) \neq 0$, as for instance in correlated UC models \citep{MorNelZi2003}, and will not affect the asymptotic results in section \ref{Ch:est}. The distributional assumptions on $\eta_t$ and $u_t$ are somewhat weaker than the assumption of Gaussian white noise on which UC models typically rely \citep{MorNelZi2003}. They will be shown to be largely satisfied in the applications of section \ref{application}. Finally, $d > 0$ is required to separately identify $\log \beta_t$ and $u_t$.

\section{Parameter estimation, filtering and smoothing}\label{Ch:est}
In this section, the fractional filter is derived. It is a computationally simple modification of the Kalman filter that avoids the Kalman recursions for the conditional variance. The modification is necessary, as the Kalman filter becomes computationally infeasible for UC models when the dimension of the state vector is high, as for fractionally integrated processes. 
The fractional filter provides a closed-form expression for the prediction error of the UC model. 
Based on that, a conditional-sum-of-squares (CSS) estimator for the model parameters is set up. It minimizes the sum of squared prediction errors obtained from the fractional filter. The CSS estimator is shown to be consistent and asymptotically normally distributed. Given the CSS parameter estimates, the log contact rate can be estimated by the fractional filter given the full sample information. Finally, estimation of the mean and seasonal components is considered.

Under assumptions \ref{asu:1} and \ref{asu:2}, the fractional UC model is given by 
\begin{align}\label{eq:mod}
	\log \tilde Y_t = \log \beta_t + \sum_{i=1}^7 \alpha_i s_{i, t} + u_t, \qquad \log \beta_t = \mu + x_t, \qquad x_t = \Delta_+^{-d} \eta_t, \qquad t=1,...,n.
\end{align}
 Denote $\mu_0, \alpha_{1, 0},...,\alpha_{7, 0}, d_0, \sigma_{\eta, 0}^2, \sigma_{u, 0}^2$ as the true parameters of the data-generating mechanism. Leaving aside the deterministic terms for the moment, by defining $y_t = \log \tilde Y_t - \mu - \sum_{i=1}^7 \alpha_i s_{i, t}$, the stochastic part of the fractional UC model \eqref{eq:mod} is
\begin{align}\label{eq:mod2}
y_t =x_t + u_t, \qquad \qquad x_t = \Delta_+^{-d} \eta_t, \qquad t=1,...,n.
\end{align}
In the following, let $\vtheta = (d, \sigma_\eta^2, \sigma_{u}^2)' \in \mTheta$ denote the vector holding the parameters of \eqref{eq:mod2}, and let $\vtheta_0 = (d_0, \sigma_{\eta,0}^2, \sigma_{u, 0}^2)' \in \mTheta$, where $\mTheta = \mD \times \mOmega_\eta \times \mOmega_u$ denotes the parameter space with $\mD = \{d \in \mathbb{R} | 0 < d \leq d_{max}\}$ and $\mOmega_i = \{\sigma_{i}^2 \in \mathbb{R} | 0 < \sigma_i^2 < \infty\}$, $i = \eta, u$. Define $\mathcal{F}_t$ as the $\sigma$-algebra generated by $y_1,...,y_t$, and let the expected value operator $\E_\vtheta(z_t)$ of an arbitrary random variable $z_t$ denote that expectation is taken with respect to the distribution of $z_t$ given $\vtheta$, so that $\E_{\vtheta_0}(z_t) = \E(z_t)$. Furthermore, let $\mSigma^{(i, j)}$ denote the $(i, j)$-th entry of an arbitrary matrix $\mSigma$.

Estimation of the parameters $\vtheta_0$ is carried out by the CSS estimator that minimizes the sum of squared prediction errors of model \eqref{eq:mod2}. The prediction error is defined as the one-step ahead forecast error of $y_{t+1}$ given $\mathcal{F}_t$
\begin{align}\label{eqn:v}
	v_{t+1}(\vtheta) = y_{t+1} - \E_\vtheta(y_{t+1} | \mathcal{F}_{t}) =  y_{t+1} - \E_\vtheta(x_{t+1} | \mathcal{F}_{t}).
\end{align}
It depends on $\E_\vtheta(x_{t+1} | \mathcal{F}_{t})$, for which the fractional filter provides an analytical solution. The filter is introduced in the following lemma.
\begin{lemma}[Fractional filter for $x_{t+1}$ given $\mathcal{F}_t$]\label{L:Kalman}
	Under assumptions \ref{asu:1} and \ref{asu:2}
	\begin{align*}
		\E_\vtheta(x_{t+1}| \mathcal{F}_t) = \sum_{i=1}^{t} \pi_i(-d) \mSigma_{\eta_{t:1}y_{t:1}}^{(i, \cdot)} \mSigma_{y_{t:1}}^{-1} y_{t:1},
	\end{align*}
	where $y_{t:1} = (y_t,...,y_1)'$, $\eta_{t:1} = (\eta_t,...,\eta_1)'$, $ \mSigma_{\eta_{t:1}y_{t:1}} = \Cov_\vtheta(\eta_{t:1}, y_{t:1})$, and  $\mSigma_{y_{t:1}} = \Var_\vtheta(y_{t:1})$. The superscript in $\mSigma_{\eta_{t:1}y_{t:1}}^{(i, \cdot)}$ denotes the $i$-th row of the matrix, and 
	\begin{align*}
		\mSigma_{\eta_{t:1} y_{t:1}}^{(i, j)} = \begin{cases}
		\pi_{i-j}(-d)\sigma_\eta^2 & \text{if } i \geq j, \\
		0 & \text{else,}
		\end{cases} \qquad 
		\mSigma_{y_{t:1}}^{(i, j)} = \begin{cases}
		\sigma_u^2 + \sigma_\eta^2 \sum_{k=0}^{t-i} \pi_k^2(-d) & \text{if } i = j,\\
		\sigma_\eta^2 \sum_{k=0}^{t-\max(i, j)} \pi_k(-d) \pi_{k + |i - j|}(-d)& \text{else.}
		\end{cases}  
	\end{align*}
\end{lemma}
The proof is contained in appendix \ref{app_proofs}. 
As can be seen from lemma \ref{L:Kalman}, the fractional filter provides a solution for $\E_\vtheta(x_{t+1}| \mathcal{F}_t)$ that only depends on $\vtheta$ and $y_{1},...,y_t$. 
By plugging it into \eqref{eqn:v}, one has the closed-form expression for the prediction error 
\begin{align}\label{eq:prederr}
	v_{t+1}(\vtheta) = y_{t+1} - \sum_{i=1}^{t} \pi_i(-d) \mSigma_{\eta_{t:1}y_{t:1}}^{(i, \cdot)} \mSigma_{y_{t:1}}^{-1} y_{t:1}.
\end{align}
Based on \eqref{eq:prederr} the objective function of the CSS estimator for $\vtheta_0$ is set up
\begin{align}\label{eq:QML}
	\hat \vtheta = \arg \min_{\vtheta \in \mTheta} \frac{1}{n} \sum_{t=1}^n v_t^2(\vtheta).
\end{align}
Note that estimating the parameters of the fractional UC model via the CSS estimator \eqref{eq:QML} in combination with the fractional filter deviates from the methodological state space literature: There, expectation and variance of $x_{t+1}$ conditional on $\mathcal{F}_t$ are typically obtained from the Kalman filter recursions \citep[see e.g.][ch.\ 4.3]{DurKoo2012}. The resulting prediction error and its conditional variance then enter the Gaussian \mbox{(quasi-)likelihood} function that is maximized to estimate $\vtheta_0$. However, the Kalman filter becomes computationally infeasible when the dimension of the state vector is high, as for fractionally integrated processes. Thus, a computationally simpler filter is required. The fractional filter, as defined in lemma \ref{L:Kalman}, is a modification of the Kalman filter: Its solution for $\E_\vtheta(x_{t+1}| \mathcal{F}_t)$ is identical to the Kalman filter \citep[see][ch.\ 4.2]{DurKoo2012}, but it avoids the Kalman recursions for the conditional variance of $x_{t+1}$. While the conditional variance is necessary for \mbox{(quasi-)maximum} likelihood estimation, the CSS estimator only requires a closed-form expression for the prediction error, for which the fractional filter is sufficient. The objective function of the CSS estimator in \eqref{eq:QML} is of course proportional to the exponent in the conditional Gaussian (quasi-)likelihood function. However, CSS estimation is computationally much simpler due to the fractional filter. Together, the fractional filter and the CSS estimator provide a computationally feasible alternative to the Kalman filter and the \mbox{(quasi-)maximum} likelihood estimator, particularly for UC models with richer long-run dynamics. 

While the asymptotic theory of the CSS estimator is well established for ARFIMA models, see \citet{HuaRob2011} and \citet{Nie2015}, it has not yet been derived for structural UC models. To fill this gap, theorems \ref{th:cons} and \ref{th:norm} summarize the asymptotic estimation theory for the CSS estimator for fractional UC models. In addition, the finite sample properties are addressed by a Monte Carlo study in appendix \ref{MC}. 
For consistency and asymptotic normality of the CSS estimator, the moment assumptions on the shocks $\eta_t$, $u_t$ need to be strengthened.

\begin{assumption}[Higher moments of $\eta_t, u_t$]\label{asu:3}
	The conditional moments of $\eta_t$, $u_t$ (conditional on past $\eta_{t-1}, \eta_{t-2},...$, and $u_{t-1}, u_{t-2},...$) are finite up to order four and equal the unconditional moments.
\end{assumption}

\begin{theorem}[Consistency]\label{th:cons}
	Under assumptions \ref{asu:1} to \ref{asu:3} the CSS estimator $\hat \vtheta$ is consistent, $\hat \vtheta \pto \vtheta_0$ as $n \to \infty$.
\end{theorem}
The proof of theorem \ref{th:cons} is given in appendix \ref{app_proofs} and is carried out as follows: First, the model in \eqref{eq:mod2} is shown to be identified. Next, $v_t(\vtheta)$, as given in \eqref{eq:prederr}, is shown to be integrated of order $d_0 - d$, and thus is stationary for $d_0 - d < 1/2$ and nonstationary for $d_0 - d > 1/2$. As the asymptotic behavior of the objective function changes around the point $d_0 - d = 1/2$, the objective function does not uniformly converge in probability on $\mTheta$. Adopting the results of \cite{HuaRob2011} and \cite{Nie2015}, who show for ARFIMA models encompassing the reduced form of \eqref{eq:mod2} that the probability of the CSS estimator to stay in the region of the parameter space where $v_t(\vtheta)$ is nonstationary is asymptotically zero, the relevant region of $\mTheta$ asymptotically reduces to the region where $d_0 - d < 1/2$ holds. Within the relevant region of the parameter space this paper then proves weak convergence of the objective function by showing the objective function to satisfy a uniform weak law of large numbers. This yields consistency of the CSS estimator, see \citet[thm.\ 4.3]{Woo1994}.

\begin{theorem}[Asymptotic normality]\label{th:norm}
	Under assumptions \ref{asu:1} to \ref{asu:3} the CSS estimator $\hat \vtheta$ is asymptotically normally distributed, $\sqrt{n}(\hat \vtheta - \vtheta_0) \dto \mathrm{N}(0, \mOmega^{-1}_0)$ as $n \to \infty$. 
\end{theorem}

The proof of theorem \ref{th:norm} is again contained in appendix \ref{app_proofs}. Since the CSS estimator is consistent, the asymptotic distribution theory is inferred from a Taylor expansion of the score function about $\vtheta_0$. A central limit theorem is shown to hold for the score function at $\vtheta_0$, together with a uniform weak law of large numbers for the Hessian matrix. The latter allows to evaluate the Hessian matrix in the Taylor expansion of the score function at $\vtheta_0$. Thus, the asymptotic distribution of the CSS estimator, as given in theorem \ref{th:norm}, can be inferred from solving the Taylor expansion for $\sqrt{n}(\hat \vtheta - \vtheta_0)$. As usual in the state space literature, no analytical solution to the asymptotic variance of the CSS estimator can be provided. The parameters of the reduced form depend non-trivially on $\vtheta$, so that the partial derivatives of the reduced form cannot be analytically derived. However, from theorem \ref{th:norm} it follows that an estimate for the parameter covariance matrix can be obtained from the negative inverse of the Hessian matrix computed in the numerical optimization.

Estimation of the latent component $x_t$ in \eqref{eq:mod2} is considered next. In line with the methodological literature on state space models, $x_t$ is estimated by plugging the CSS estimates $\hat \vtheta$ into the projection
\begin{align}\label{KS:1}
	x_{t|n}(\vtheta) = \Cov_\vtheta(x_t, y_{n:1})\Var_\vtheta(y_{n:1})^{-1}y_{n:1}=
	\sum_{i=0}^{t-1} \pi_i(-d) \mSigma_{\eta_{t:1}y_{n:1}}^{(i, \cdot)} \mSigma_{y_{n:1}}^{-1} y_{n:1},
\end{align}
where $y_{n:1} = (y_n,...,y_1)'$, $\mSigma_{\eta_{t:1}y_{n:1}} = \Cov_\vtheta(\eta_{t:1}, y_{n:1})$, and  $\mSigma_{y_{n:1}} = \Var_\vtheta(y_{n:1})$. The superscript in $\mSigma_{\eta_{t:1}y_{n:1}}^{(i, \cdot)}$ denotes the $i$-th row of the matrix, and 
\begin{align*}
	\mSigma_{\eta_{t:1} y_{n:1}}^{(i, j)} = \begin{cases}
	\pi_{n-t+i-j}(-d)\sigma_\eta^2 & \text{if } n-j \geq t- i, \\
	0 & \text{else,}
	\end{cases} 
\end{align*}
while the entries of $\mSigma_{y_{n:1}}$ follow from lemma \ref{L:Kalman} by setting $t=n$. For $\vtheta = \vtheta_0$, $x_{t|n}(\vtheta_0)$ is the minimum variance linear unbiased estimator given $y_1,...,y_n$, see \citet[lemma 2]{DurKoo2012}. Due to theorem \ref{th:cons}, this property holds asymptotically for $x_{t|n}(\hat \vtheta)$. Note that \eqref{KS:1} is identical to the Kalman smoother, see \citet[ch.\ 4.4]{DurKoo2012}. However, \eqref{KS:1} is computationally much simpler, as it avoids the computationally intensive Kalman recursions for the conditional variance. In line with lemma \ref{L:Kalman}, \eqref{KS:1} is the fractional filter for $x_t$ given $\mathcal{F}_n$. 
	
Finally, estimation of the seasonal components $\alpha_{i,0}$, $i=1,...,7$, and $\mu_0$ is considered. Theoretically, all parameters of the model in \eqref{eq:mod} could be estimated jointly by the CSS estimator. But as \cite{TscWebWe2013b} explain, including deterministic terms in the optimization can lead to poor results in finite samples for fractionally integrated processes, particularly when $d_0$ is close to unity, as the deterministic terms suffer from poor identification. They provide simulation evidence and a line of reasoning explaining why the following two-step estimator is more robust: In the first step, the integration order $d_0$ is estimated using the exact local Whittle estimator of \cite{Shi2010}, which allows for unknown deterministic terms and yields $\hat d_{EW}$. Based on $\hat d_{EW}$, the deterministic terms $\mu_0, \alpha_{0,1},...,\alpha_{0,7}$ in
\begin{align}\label{det1}
	\Delta_+^{\hat d_{EW}} \log \tilde Y_t = \Delta_+^{\hat d_{EW}} \mu + \sum_{i=1}^7 \alpha_i \Delta_+^{\hat d_{EW}} s_{i, t} + error_t, 
\end{align}
are estimated by ordinary least squares. In the second step, the objective function of the CSS estimator in \eqref{eq:QML} is minimized for the adjusted $\log \tilde Y_t  - \hat \mu -  \sum_{i=1}^7 \hat\alpha_i s_{i, t}$.

As an alternative to \eqref{det1}, one could also eliminate the seasonal components by averaging over seven neighboring observations, as $\sum_{i=1}^7 \alpha_{i, 0} s_{i, t} = 0.$ The intercept in
\begin{align}\label{det2}
	\Delta_+^{\hat d_{EW}}\frac{1}{7}\sum_{i=0}^{6} \log \tilde Y_{t-q+i} = \Delta_+^{\hat d_{EW}} \mu + error_t, \qquad 0 \leq q \leq 6,
\end{align}
could then be estimated by ordinary least squares. $q$ determines whether averages are calculated solely based on past data ($q = 6$), based on centered data around $t$ ($q=3$), or based on future data ($q=0$). While the second approach does not require to estimate $\alpha_{1,0},...,\alpha_{7,0}$, averaging over seven days smooths out potential kinks in the contact rate which is problematic. Furthermore, averaging may pollute the estimates of $x_t$ and induce spurious long memory. Finally, the choice of $q$ is not trivial: While for forecasting purposes $q=6$ is adequate, choosing $q=0$ is likely to account best for the delay in reporting of case numbers, and obviously $q=3$ may be a good compromise between the two options. In the applications $\mu_0, \alpha_{0,1},...,\alpha_{0,7}$ will be estimated via \eqref{det1}. 
\section{Empirical results}\label{application}
In this section, estimation results for the time-varying contact rate $\beta_t$ are presented for Canada, Germany, Italy, and the United States. The underlying data on confirmed, recovered, and deceased cases stems from the JHU CSSE. As  in \cite{LeeLiaSe2021} and \cite{LiuMooSc2021}, $t=1$ is set once the number of cumulative cases reaches $100$. Prior smoothing as suggested by \cite{LeeLiaSe2021} and \cite{LiuMooSc2021}, who use one-sided three-day rolling averages to smooth the data, is avoided, as this likely pollutes the kinks in the contact rate that occur due to containment measures.\footnote{\cite{LiuMooSc2021} argue that one-sided three-day rolling averages smooth out noise generated by the timing of the reporting. However, the opposite should be the case, as a one-sided smoothing shifts case numbers from past to present, while a delay in reporting shifts case numbers from present to the future. To fix the latter, a forward-looking filter is required, not a backward-looking one, see the discussion at the end of section \ref{Ch:est}.}

Instead of smoothing out seasonality, the data is adjusted for weekly seasonal patterns as described at the end of section \ref{Ch:est} using \eqref{det1}. The bandwidth for the exact local Whittle estimator in \eqref{det1} is set to $m=\lfloor n^{0.65} \rfloor$, which is justified by the Monte Carlo study in appendix \ref{MC}. Based on the seasonally adjusted data, the parameters $\vtheta_0$ are estimated via the CSS estimator \eqref{eq:QML}, where $100$ combinations of starting values for $\vtheta_0$ are drawn from uniform distributions with appropriate support (in particular $d \in [0.5, 2]$) to avoid convergence to a local optimum. However, due to the parsimonious parametrization of the model, all combinations of starting values converged to virtually identical optima, implying that the procedure is robust to the choice of starting values. Plugging the CSS estimates into \eqref{KS:1}, together with $\hat \mu$ in \eqref{det1}, yields the log contact rate estimate $\log \hat \beta_t$.

The average infected period is required for $\mathcal{R}_t$ and is estimated by solving \eqref{eqn:2} for $\gamma$ and taking the average
\begin{align}\label{eqn:gamma}
	\hat \gamma = \frac{1}{n-1}\sum_{t=2}^{n} \left[ \hat \beta_t S_{t-1} - \frac{\Delta I_t}{I_{t-1}}\right]. 
\end{align}
This reflects that the definition of recovered varies over the countries under study, particularly as non-hospitalized persons are typically assumed to have recovered $h$ days after they tested positive, and $h$ varies over the countries under study. The choice of $h$ proportionally affects the number of currently infected $I_t$, and thus $\beta_t$ is inversely proportional to $h$ by \eqref{eqn:5}. Consequently, only the dynamics of $\beta_t$ should be compared over the different countries, not the absolute numbers. In contrast, the reproduction rate $\mathcal{R}_t = \beta_t / \gamma$ accounts for the different $h$ when $\gamma$ is estimated via \eqref{eqn:gamma}. If instead $\gamma = 1/18$ is fixed as in \cite{LeeLiaSe2021}, the dependence on $h$ is not resolved and countries with a higher $h$ will exhibit a smaller reproduction rate by construction. This is precisely the reason for the implausible estimates for $\mathcal{R}_t$ in \cite{LeeLiaSe2021}, and is solved by accounting for different $h$ via \eqref{eqn:gamma}.

Results are reported for Canada, Germany, and Italy in subsection \ref{subsec:1}. They are selected as they are all members of the G7 and have implemented containment measures of different strength, duration, and at different points in time, thus making a comparison interesting. For the selected countries, there exist reliable data on confirmed, recovered, and deceased cases provided by the JHU CSSE. As will be shown, the latter is not the case for the US, where data on recovered subjects suffers heavily from under-reporting, yielding a severe downward-bias for the estimated contact rate and the resulting reproduction rate $\mathcal{R}_{t}$ as reported by \cite{LeeLiaSe2021}. The problem is fixed by an assumption on the average duration of an infection, and results for the US are presented in subsection \ref{subsec:2}. As will become apparent there, the fractional filter is quite robust to under-reporting of recovered cases. To monitor the pandemic in real-time, subsection \ref{subsec:3} examines the precision of the fractional filter at the end of the sample.

\subsection{Canada, Germany, and Italy}\label{subsec:1}
For Canada, figure \ref{fig:1} sketches the estimated log contact rate and the resulting reproduction rate $\hat{\mathcal{R}}_{t} = \hat \beta_t/\hat\gamma$ in the first row. The average duration of an infection is estimated to be $1/\hat \gamma = 18.29$ days.
The second row of figure \ref{fig:1} displays the estimated prediction error $\vv_t(\hat \vtheta)$ and its estimated autocorrelation function. 

\begin{figure}[p]
	\includegraphics[width=\linewidth]{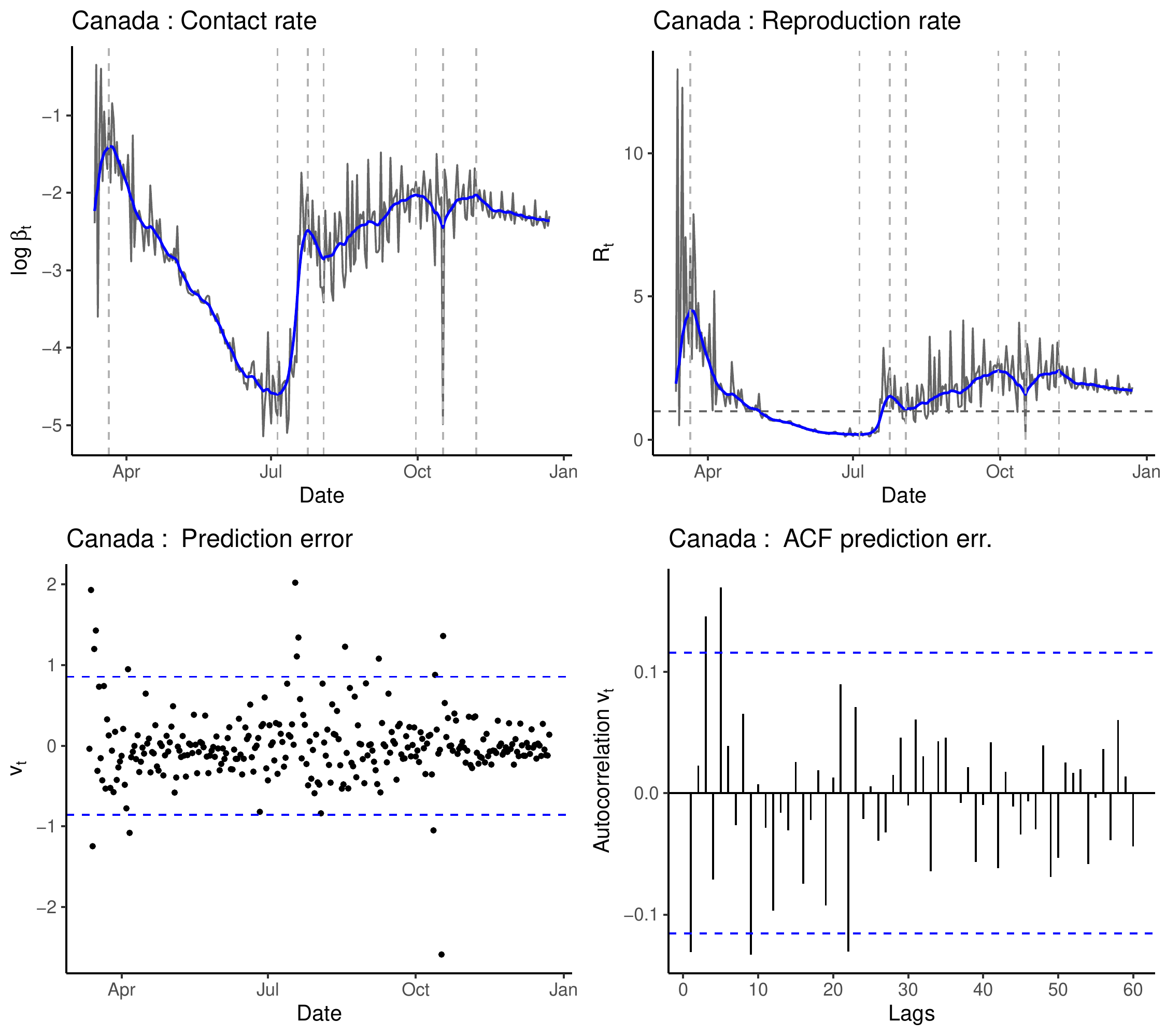}
	\caption[]{Estimation results for Canada. The top-left panel displays the estimated contact rate $\log \hat \beta_t$ in blue together with the observable $\log \tilde Y_t$ in gray. The top-right panel shows the estimated reproduction rate $\hat{\mathcal{R}}_{t} = \hat \beta_t/\hat\gamma$ in blue together with $\tilde Y_t/\hat \gamma$ in gray. The dashed horizontal line corresponds to $\mathcal{R} = 1$. The dashed vertical lines correspond to turning points of the contact rate. The bottom-left panel shows the estimated prediction error $\vv_t(\hat \vtheta)$ in \eqref{eq:prederr} together with two standard deviations in blue, dashed. The bottom-right panel sketches the estimated autocorrelation function of the prediction error $\vv_t(\hat \vtheta)$ together with a $95\%$ confidence interval. }
	\label{fig:1}
\end{figure}%
Based on the top-left panel of figure \ref{fig:1}, several turning points of the contact rate can be identified using a simple algorithm that defines a minimum (maximum) whenever the contact rate $\beta_t$ at $t$ is smaller (greater) than all $\beta_{t+1},...,\beta_{t+10}$, the contact rates of the next ten days. 
These periods correspond to several policy regimes characterized by the strengthening and easing of containment measures. While a small selection of policy measures is presented below, a detailed overview is given by \citet[][]{McCSmiAn2020}.
\begin{enumerate}[noitemsep]
	\item \textbf{March 13 -- March 21:} The contact rate increases and peaks on March 21. As a reaction, several provinces and territories declare the state of emergency between March 13 and March 22, impose gathering bans, close schools, universities, and businesses, and cancel mass events, among others.
	\item \textbf{March 22 -- July 5:} After the implementation of containment measures the contact rate decreases continuously. While additional containment measures such as travel restrictions are implemented in April, several provinces and territories start to step-wise relax their restrictions in May. 
	On July 1, Canada's national summer holidays begin. 
	\item \textbf{July 6 -- July 24:} During the first half of the summer holidays the contact rate increases sharply. The reproduction rate increases above unity on July 20.
	\item \textbf{July 25 -- August 3:} The contact rate slightly decreases, while the reproduction rate remains above unity.
	\item \textbf{August 4 -- September 30:} After a short phase of reduction, contact and reproduction rate start to increase again, while schools re-open on September 8. Canada's prime minister Trudeau says the second wave of COVID-19 is already underway.
	\item \textbf{October 1 -- October 17:} Contact and reproduction rate slightly decrease. Thanksgiving takes place on October 12.
	\item \textbf{October 18 -- November 7:} After Thanksgiving, contact and reproduction rate increase slightly. On November 3rd, Ontario introduces an incidence-based system for when to tighten containment measures.
	\item \textbf{November 8 -- December 23:} Contact and reproduction rate exhibit a slight but steady decrease. Reproduction rate remains above unity.
\end{enumerate}
The estimated contact rate and the resulting reproduction rate are well in line with the chronology of policy interventions. In particular, the fractional filter allows to identify turning points of the contact rate that are not visible from the raw data that is plotted in gray color in figure \ref{fig:1}. 

The two graphs at the bottom of figure \ref{fig:1} illustrate how well the Canadian data fits the model assumptions. Assumptions \ref{asu:1} and \ref{asu:2} assume the measurement error $u_t$ and the log contact rate shock $\eta_t$ to be homoscedastic white noise processes. Since $\hat \vtheta$ is consistent, see theorem \ref{th:cons}, by \eqref{eq:a1} the prediction error 
$\vv_t(\hat \vtheta)$ becomes a white noise process as $n \to \infty$ if assumptions \ref{asu:1} and \ref{asu:2} hold. While some outliers exist, about $95\%$ of the prediction errors lie within two standard deviations, as the bottom-left panel illustrates. The bottom-right panel shows that there is not much autocorrelation left in the prediction error. Given the parsimonious parametrization of the fractional UC model, this is surprising. The two panels at the bottom of figure \ref{fig:1} thus substantiate that the dynamics of the log contact rate are well captured by a fractionally integrated process.

The estimated integration order is $\hat d = 1.2166$, which implies that a unit shock on the contact rate growth $\Delta \log \beta_t$ retains $21.66\%$ of its impact in $t+1$, $13.17\%$ in $t+2$, and $9.73\%$ in $t+3$. After one week, the impact is still $5.10\%$, after two weeks $2.98\%$, and after three weeks $2.17\%$, which is due to the strong persistence of fractionally integrated processes, see assumption \ref{asu:2} for the formula for $\pi_i(d-1)$. The slow decay may very well describe the persistent impact of past information shocks on today's social behavior.

For Germany, figure \ref{fig:2} plots the empirical results.
The average infected period is estimated to be $1/\hat \gamma = 21.27$ days, which is slightly greater compared to Canada and likely results from different algorithms to estimate the number of recovered individuals. The integration order estimate is of similar size as for Canada ($\hat d = 1.2693$, see table \ref{ta:1}). 
\begin{figure}[p]
	\includegraphics[width=\linewidth]{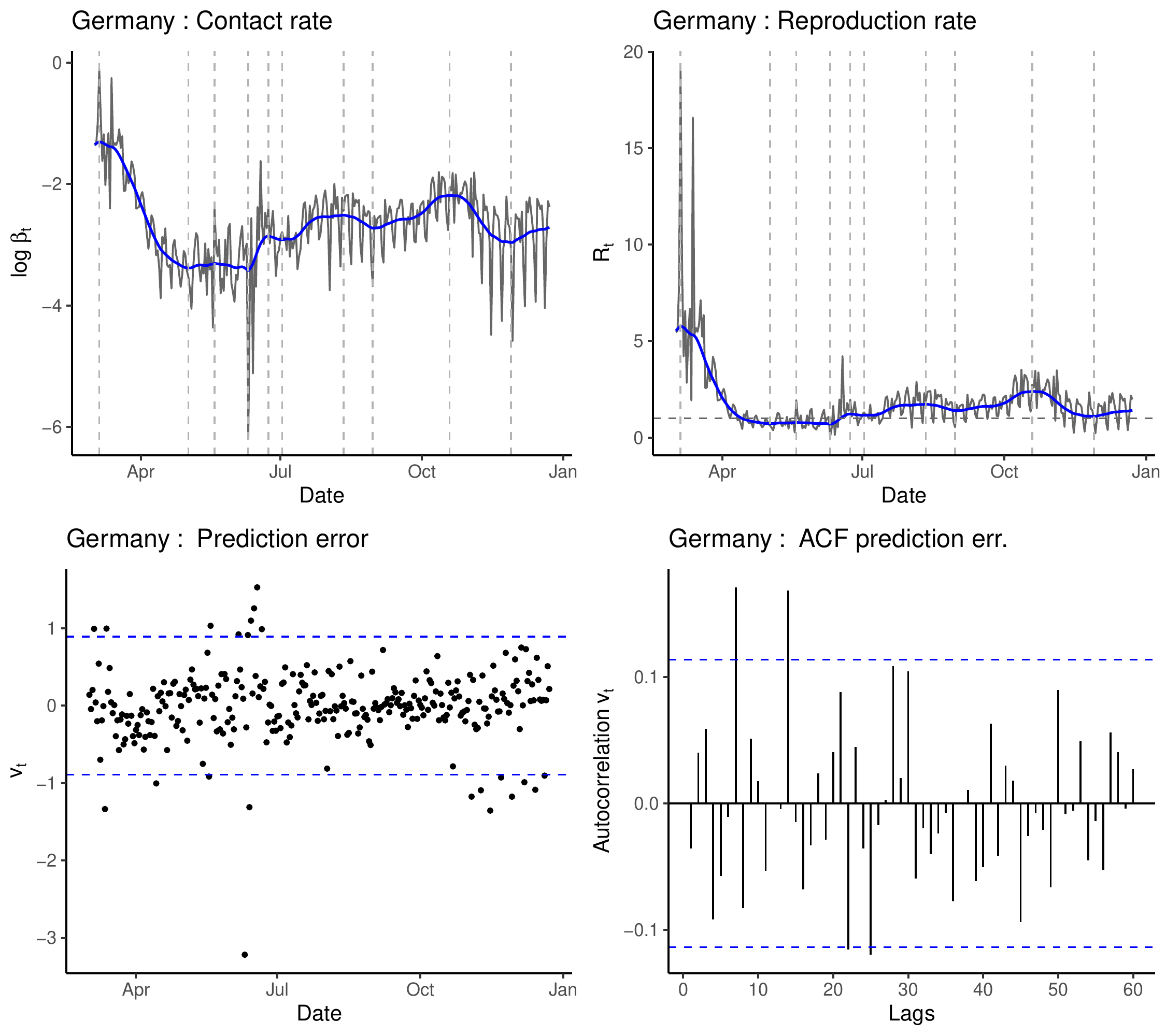}
	\caption[]{Estimation results for Germany. For a description see figure \ref{fig:1}.}
	\label{fig:2}
\end{figure}%
From the top-left panel of figure \ref{fig:2} the following contact regimes can be identified:
\begin{enumerate}[noitemsep]
	\item \textbf{March 2 -- March 5:} The contact rate starts at a comparably high level on March 2, likely caused by Carnival celebrations and ski tourism during Germany's winter holidays \citep{FelHinCh2020}. It peaks on March 5. 
	\item \textbf{March 6 -- May 2:} A slight decrease at the beginning of March turns into a sharp decrease around March 16. Measures to contain the spread of the virus, such as school closings and event cancellations, are implemented from March 13 on. From March 22 on, gatherings of more than two people are prohibited and several businesses are closed \citep{HarWaeWe2020}. At the end of April, schools and businesses partly re-opened.
	\item \textbf{May 3 -- May 19:} Contact and reproduction rate slightly increase. Gathering restrictions are relaxed and most businesses are allowed to re-open on May 6.
	\item \textbf{May 20 -- June 10:} Contact and reproduction rate slightly decrease.
	\item \textbf{June 11 -- June 23:} A short but strong increase is caused, among others, by a massive outbreak in a meat factory \citep{BBC2020}. 
	\item \textbf{June 24 -- July 2:} A slight decline follows. On June 29, summer holidays begin in Germany's largest state.
	\item \textbf{July 3 -- August 11:} During the summer holidays, Germany experiences a further increase in its contact and reproduction rate.
	\item \textbf{August 12 -- August 30:} Contact and reproduction rate decrease.
	\item \textbf{August 31 -- October 19:} A slight increase of the contact rate is followed by a strong increase at the end of September. On October 15 stricter rules for hotspots are implemented, including mask obligations, contact restrictions, and curfews \citep{DW2020}.  
	\item \textbf{October 20 -- November 28:} Contact and reproduction rate decrease, but the latter remains above unity. On October 28 a `lockdown light' is announced for Germany, inducing gathering restrictions and business closings among others \citep{BBC2020b}.
	\item \textbf{November 29 -- December 23:} While contact and reproduction rate slightly increase, the government announces a tightening of its lockdown measures on December 13 \citep{BBC2020c}.
\end{enumerate}
Similar to Canada, the two panels at the bottom of figure \ref{fig:2} indicate that the model assumptions are largely satisfied, despite little remaining autocorrelation in the prediction error.

For Italy, a slight adjustment of the data is required, as $\Delta C_t = 0$ on June 19 and thus $\log \tilde Y_t$ is not defined, see \eqref{eqn:5}. To adjust the single observation, averages from the neighboring observations are used (i.e.\ $\Delta C_t = 1/3 (\Delta C_{t-1} + \Delta C_{t+1})$), while the adjusted cases in $t-1$ and $t+1$ will equal $2/3$ of the reported cases. Thus, the cumulated number of reported cases is unaffected by the adjustment. 

The empirical results are displayed in figure \ref{fig:3}. An estimate for the average infected period is $1/\hat \gamma = 35.92$ days, which is significantly higher than for Canada and Germany. The discussion below \eqref{eqn:gamma} gives an explanation for the high variation in $\gamma$ over the different countries. 
The estimated integration order is $\hat d =1.4304$ (see table \ref{ta:1}), which is somewhat greater than the estimates for Germany and Canada, implying that a shock $\eta_t$ in Italy will yield a more persistent effect on the contact rate. This may be explained by the severity of the pandemic in Italy in spring 2020, which likely had a long-lasting impact on social behavior.  
Based on the top-left panel of figure \ref{fig:3} the following regimes and turning points are visible:

\begin{figure}[p]
	\includegraphics[width=\linewidth]{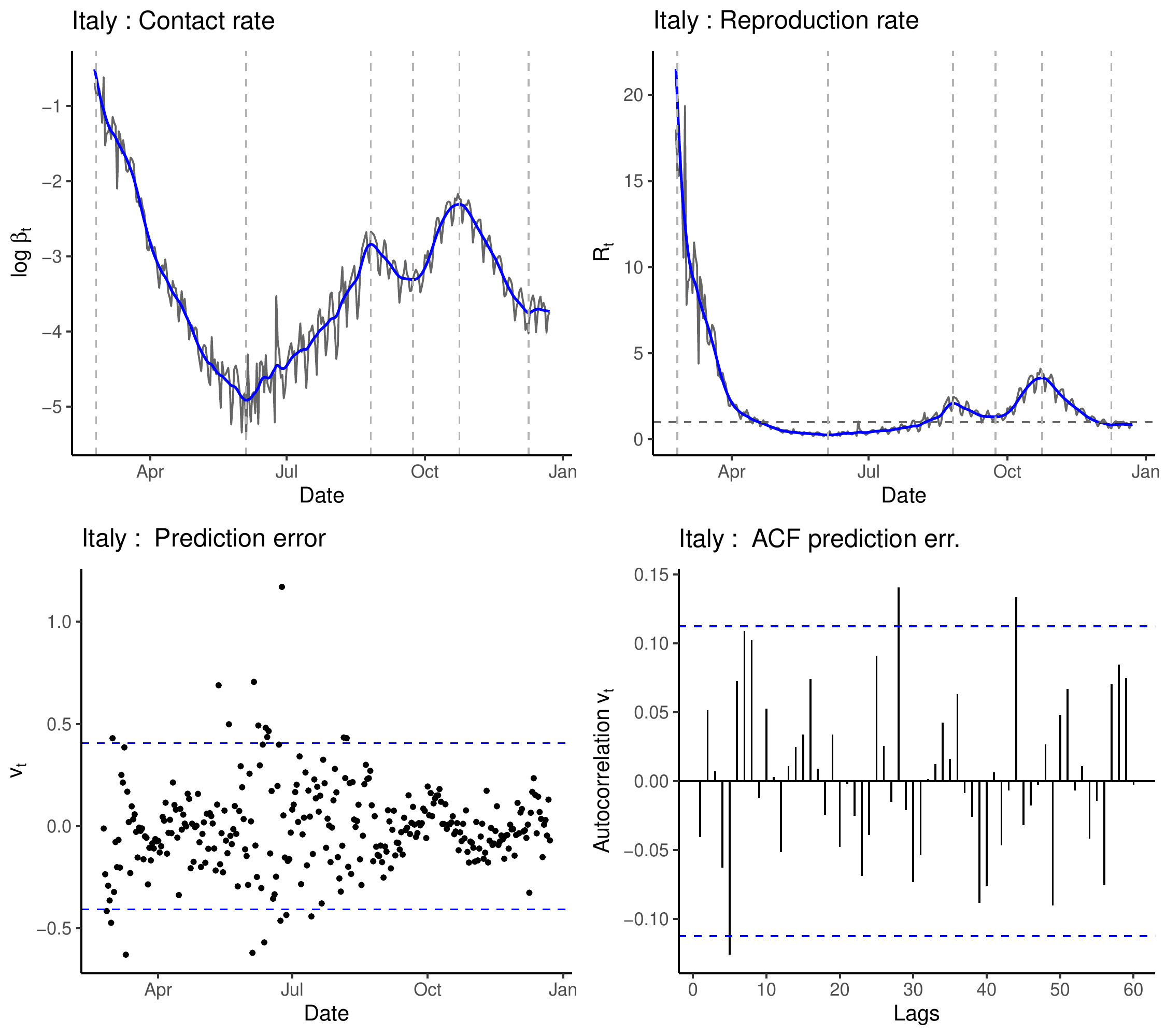}
	\caption[]{Estimation results for Italy. For a description see figure \ref{fig:1}.}
	\label{fig:3}
\end{figure}%
\begin{enumerate}[noitemsep]
	\item \textbf{February 24 -- February 25:} Due to several clusters in northern Italy more than 100 cumulative cases are counted on February 23. While the contact rate peaks on February 25,
	epicenters in northern Italy effectively went into lockdown on February 22 \citep{SogScoOd2020}. 
	\item \textbf{February 26 -- June 4:} The contact rate exhibits the strongest decline among all countries under study. During March, the government announces several containment measures such as school closings, halting all non-essential businesses, and tight regulations on free movement \citep{SogScoOd2020}. During May, the lockdown is lifted gradually. It effectively ends on June 3 \citep{BBC2020d}.
	\item \textbf{June 5 -- August 26:} The contact rate exhibits a long and steep increase.  
	\item \textbf{August 27 -- September 23:} A short decrease of the contact rate follows.
	\item \textbf{September 24 -- October 24:} The contact rate again increases and reaches a level as high as at the end of March. Gatherings, restaurants, sports and school activities are again restricted \citep{DW2020b}.
	\item \textbf{October 25 -- December 9:} The contact rate strongly declines, while additional restrictions on bars and restaurants are implemented \citep{DW2020c}.
	\item \textbf{December 10 -- December 23:} A minor increase in the contact rate is visible the days before Christmas.
\end{enumerate}
The two panels at the bottom of figure \ref{fig:3} are similar to Canada and Germany, and indicate that the prediction errors are rather homoscedastic, although outliers exist, and little autocorrelation is left.

\subsection{United States}\label{subsec:2}
The US is treated separately, since data on recovered cases reported by the JHU CSSE seem heavily downward-biased. To see this, consider the difference between lagged cumulative confirmed, cumulative recovered, and cumulative deceased cases for different lags $h$ 
\begin{align}\label{diff}
	C_{t-h} - R_t - D_t. 
\end{align}
For $h=0$, \eqref{diff} measures the number of currently infected subjects. For small $h$, \eqref{diff} should be positive, as it takes some time for the infected subjects to either recover or die. As $h$ increases, \eqref{diff} should turn negative, as an increasing number of subjects contained in the cumulative cases $C_{t-h}$ and subjects infected between $t-h$ and $t$ (and thus contained in $C_t - C_{t-h}$) either recover or die. The turning point, denoted by $\bar h$, should be close to the average infected period $1/\gamma$, as long as new confirmed cases between $t-h$ and $t$, i.e.\  $C_t - C_{t-h}$, do not explode. If they do, then $\bar h$ should be smaller than $1/\gamma$, as outflows from $C_t - C_{t-h}$ disproportionally increase $R_t$ and $D_t$.
\begin{figure}[h]
	\includegraphics[width=\linewidth]{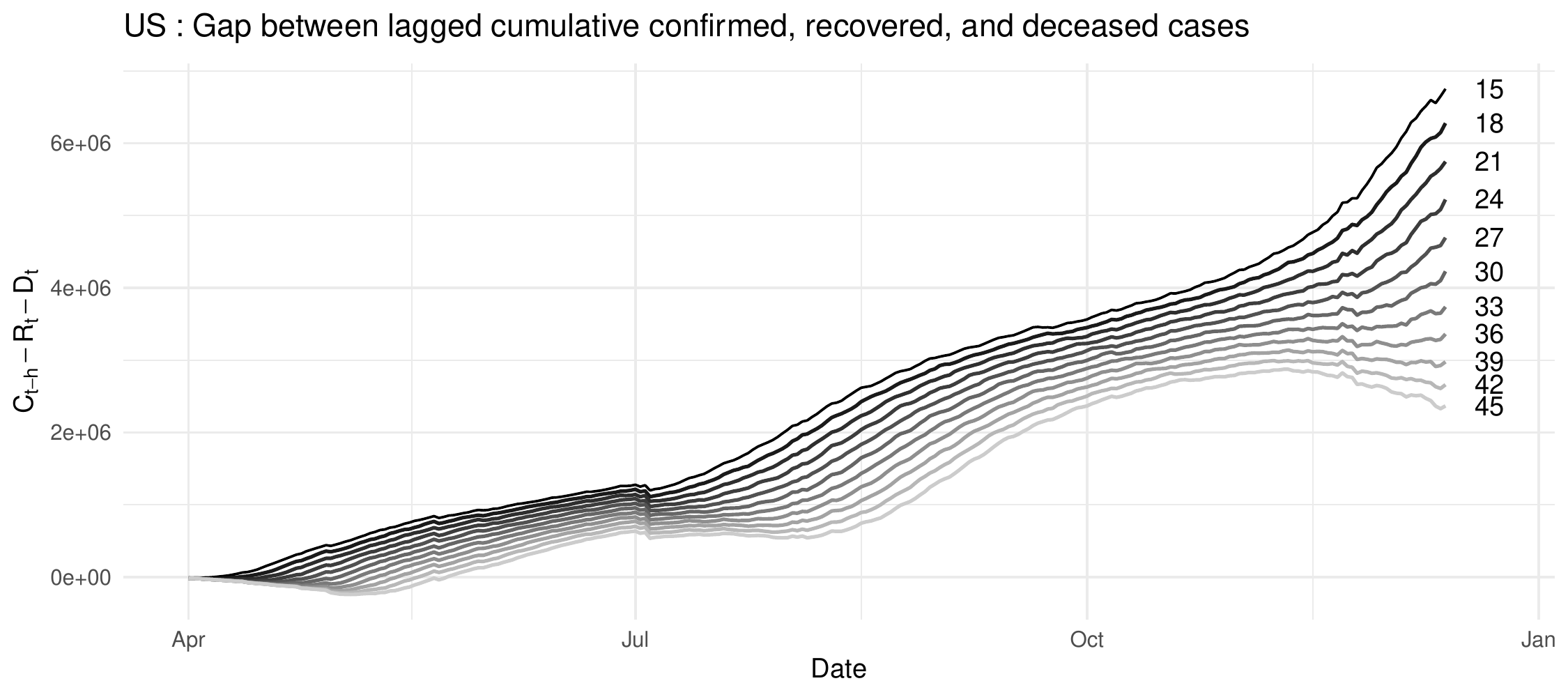}
	\caption[]{Difference between lagged cumulative confirmed, recovered, and deceased cases $C_{t-h} - R_t - D_t$ (in case numbers) for $h=15, 18,...,42,45$.}
	\label{fig:5}
\end{figure}%

Figure \ref{fig:5} plots \eqref{diff} in case numbers for lags $h=15,18,...,42,45$. As can be seen, even after $45$ days the difference between lagged cumulative confirmed, cumulative recovered, and cumulative deceased cases is predominantly positive. This is at odds with the average infected periods for Canada, Germany, and Italy as found in subsection \ref{subsec:1}, and indicates that data on $R_t$, $D_t$ may suffer from under-reporting. As stated by the JHU CSSE, data on recovered cases are based on local media reports as well as on state reporting when available. 
US state-level recovered cases stem from the COVID Tracking Project (www.covidtracking.com). \citet{CTR2020} recently pointed out that several states and territories, including California and Florida, do not report data on recovered cases, which may explain the downward-bias. In addition, definitions of recovered cases differ considerably across states, and the majority of the states consider a case as recovered if a certain number of days (generally between 10 and 30) after a positive test result or symptom onset have passed and the patient has not died. 

Since no reliable data on recovered cases is available for the US, an approximation is required. In the following, it will be assumed that individuals either recover or die $\bar{h} = 21$ days after they tested positive, $C_{t-21} = D_t + R_t$. The assumption is justified as follows: First, it is similar to the average infected period estimated for Germany and more conservative than the estimate for Canada. And second, it is centered in the range of definitions for recovered individuals by the federal states. In addition, estimates for $\bar{h} = 18$ and $\bar h = 24$ days are presented, which gives a reasonable interval for the contact rate. 

\begin{figure}[p]
	\includegraphics[width=\linewidth]{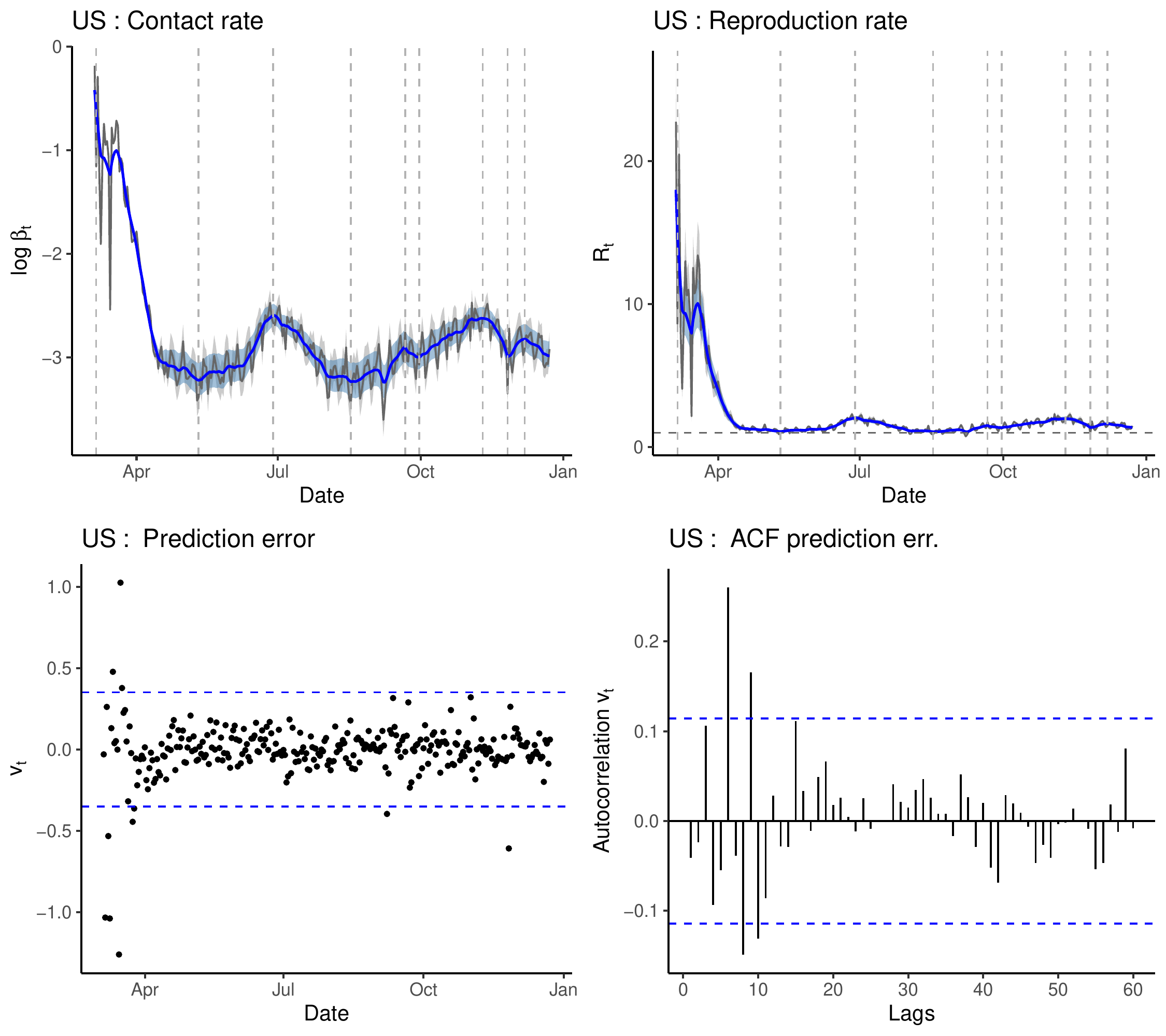}
	\caption[]{Estimation results for the United States for $\bar h=21$. Shaded areas correspond to $\bar h=18$ and $\bar h=24$. For a description see figure \ref{fig:1}.}
	\label{fig:4}
\end{figure}%

Under the assumption of $\bar h = 21$, the estimated integration order equals $\hat d = 1.2499$ (see table \ref{ta:1}) and is very similar to the results in subsection \ref{subsec:1}. Estimates for contact and reproduction rate are visualized in figure \ref{fig:4} and again allow to decompose the chronology of the pandemic into different regimes: 
\begin{enumerate}[noitemsep]
	\item \textbf{March 5 -- March 6:} Contact and reproduction rate peak at the initial stage of the epidemic. 
	\item \textbf{March 7 -- May 11:} Contact and reproduction rate steadily decrease, despite a small blip on March 19. The national state of emergency is declared on March 13, and several containment measures such as business closings and stay-at-home orders are implemented mainly during the second half of March. Depending on the state, businesses re-open from April 20 on. More than half of the states have opened businesses on May 7  \citep{CheKasSc2021}.
	\item \textbf{May 12 -- June 28:} As containment measures are relaxed, contact and reproduction rate increase slightly during the second half of May and experience a strong increase during June. 
	\item \textbf{June 29 -- August 17:} Contact and reproduction rate decline, reaching the level of May.
	\item \textbf{August 18 -- September 21:} A slight increase is visible.
	\item \textbf{September 22 -- September 30:} A short decrease follows.
	\item \textbf{October 1 -- November 10:} Contact and reproduction rate exhibit a steady increase, reaching the June peak. Regional containment measures are implemented at the beginning of November.
	\item \textbf{November 11 -- November 26:} Contact and reproduction rate decrease.
	\item \textbf{November 27 -- December 7:} After Thanksgiving, contact and reproduction rate experience a short increase.
	\item \textbf{December 8 -- December 23:} A decrease is visible from December 8 on. 
\end{enumerate}
As can be seen from figure \ref{fig:4}, estimates for the contact rate are rather robust to the choice of $\bar h$. They are slightly greater for $\bar h = 18$, as the number of currently infected $I_t$ is smaller by construction and thus additional contacts are required to explain new confirmed cases, while they are slightly smaller for $\bar h = 24$ exactly for the opposite reason. The estimated reproduction rate is virtually identical among the three scenarios, as it is normalized by the average infected period $\hat{\mathcal{R}}_{ t} = \hat \beta_t / \hat \gamma$. 
For $\bar h=21$, the two panels at the bottom of figure \ref{fig:4} indicate that the model assumptions are largely satisfied, despite some weak correlation in the prediction errors. The plots are very similar for $\bar h=18$ and $\bar h=24$ and thus not shown. The reproduction rate $\hat{\mathcal{R}}_{t}$ is greater than unity during the whole sample, which contradicts the results of \cite{LeeLiaSe2021} who rely on the downward-biased data on recovered subjects from the JHU CSSE while fixing the average infected period to be $1/\gamma = 18$.

\subsection{Monitoring the current state of the pandemic}\label{subsec:3}
This subsection investigates the end-of-sample properties of the fractional filter for real-time estimation of the contact rate. Reliable contact rate estimates at the current frontier of the data would allow to real-time monitor the state of the pandemic and can serve as a surveillance measure for future outbreaks. Based on reliable real-time estimates for the contact rate, policy rules can be implemented to prevent an exponential growth of case numbers. Acting early reduces economic and social costs of containment measures, and consequently a well-designed policy rule will be beneficial, given that the fractional filter yields a reliable estimate for the current level of the contact rate. Drawing inference on the latter is the focus of this subsection.

In detail, real-time monitoring is simulated by truncating the sample at a certain point $t$, $r \leq t \leq n$, where $r$ is the minimum sample size for the CSS estimator to produce reasonable estimates. The parameters $\vtheta_0, \mu_0, \alpha_{1,0},...,\alpha_{7, 0}$ of \eqref{eq:mod} are then estimated as described in section \ref{Ch:est} using the information available at time $t$, $\mathcal{F}_t$, and the resulting parameter estimates are denoted as $\hat \vtheta^{(t)}$, $\hat \mu^{(t)}$, etc. To take into account reporting lags, and to be robust against outliers at the end of the sample, a little backward-smoothing is allowed by reporting the smoothed estimate for the log contact rate at period $t-3$ given the information available at period $t$. From \eqref{KS:1}, the smoothed estimates are
\begin{align}\label{eq:beta_mon}
	\log \hat \beta_{t-3|t} = \hat \mu^{(t)} + x_{t-3|t}(\hat \vtheta^{(t)}).
\end{align}
As \eqref{eq:beta_mon} only depends on information available at $t$, it mimics the situation of a policy maker at $t$ and can be used to draw inference on the monitoring properties of the fractional filter at time $t$. Based on $\hat \beta_{t-3|t}$, policy rules to prevent an exponential spread of the virus can be designed. Such rules could, for instance, define a threshold for $\hat{\mathcal{R}}_{t-3}$ at which additional containment measures are implemented. As the threshold should naturally depend on the number of currently infected, current hospital capacities, and other parameters, the precise design of such a policy rule is left to the experts, and only a primitive policy rule will be introduced later for illustrative purposes. 

The reliability of real-time estimates for the contact rate is assessed by the following experiment: First, an estimation sample that consists of information available until May 31 is defined, for which $\vtheta_0, \mu_0, \alpha_{1,0},...,\alpha_{7, 0}$ are estimated. It consists of at least 80 observations, which is considered as a reasonable sample size for the estimation sample. 
Based on these estimates, $\log \hat \beta_{r-3|r}$ is obtained as described above. In a second step, information available on June 1 is added to the sample and parameter estimates are updated using the $\hat \vtheta^{(r)}$ from the estimation sample as starting values for the CSS estimator, which gives $\hat \vtheta^{(r+1)}$. As before, the estimate for $\log \hat \beta_{r-2|r+1}$ is stored. The procedure repeats for all $t$, $r < t \leq n$, where in every step $t$ the CSS estimator is initialized by $\hat \vtheta^{(t-1)}$. The resulting real-time estimates for the contact rate are then compared to those of subsections \ref{subsec:1} and \ref{subsec:2} to draw inference on their reliability. In addition, a primitive policy rule is introduced. It assumes governments to take action as soon as $\hat{\mathcal{R}}_{t-3} = \hat{\beta}_{t-3|t} / \hat \gamma > 1.2$. The latter is motivated by the observation that preventing an exponential propagation (i.e.\ $\mathcal{R}_{ t-3} > 1$) is desirable, and a margin of $0.2$ is included to be robust against outliers. Finally, the real-time contact rate estimates are compared to a rolling seven-day average 
\begin{align}\label{benchmark}
	\log \hat \beta^{benchmark}_{t-3|t} = \frac{1}{n}\sum_{i=0}^{6} \log \tilde Y_{t-i},
\end{align}
which includes three forward-looking observations and should smooth out the seasonality. 

The real-time experiment considered in this paper deviates from \cite{LeeLiaSe2021}, who suggest to monitor the current state of the pandemic by fitting a linear time trend with structural breaks to $\log \tilde Y_t$. \cite{LeeLiaSe2021} evaluate the monitoring properties of their contact rate estimate ex-post, using all information available in their sample. Consequently, their estimates at point $t$ depend on information that was not available to policy makers at period $t$ whenever $t < n$. As structural breaks are not well identified at the end of the sample, recent changes in the contact rate cannot be expected to be found by the estimator of \cite{LeeLiaSe2021}. 


\begin{figure}[p]
	\includegraphics[width=\linewidth]{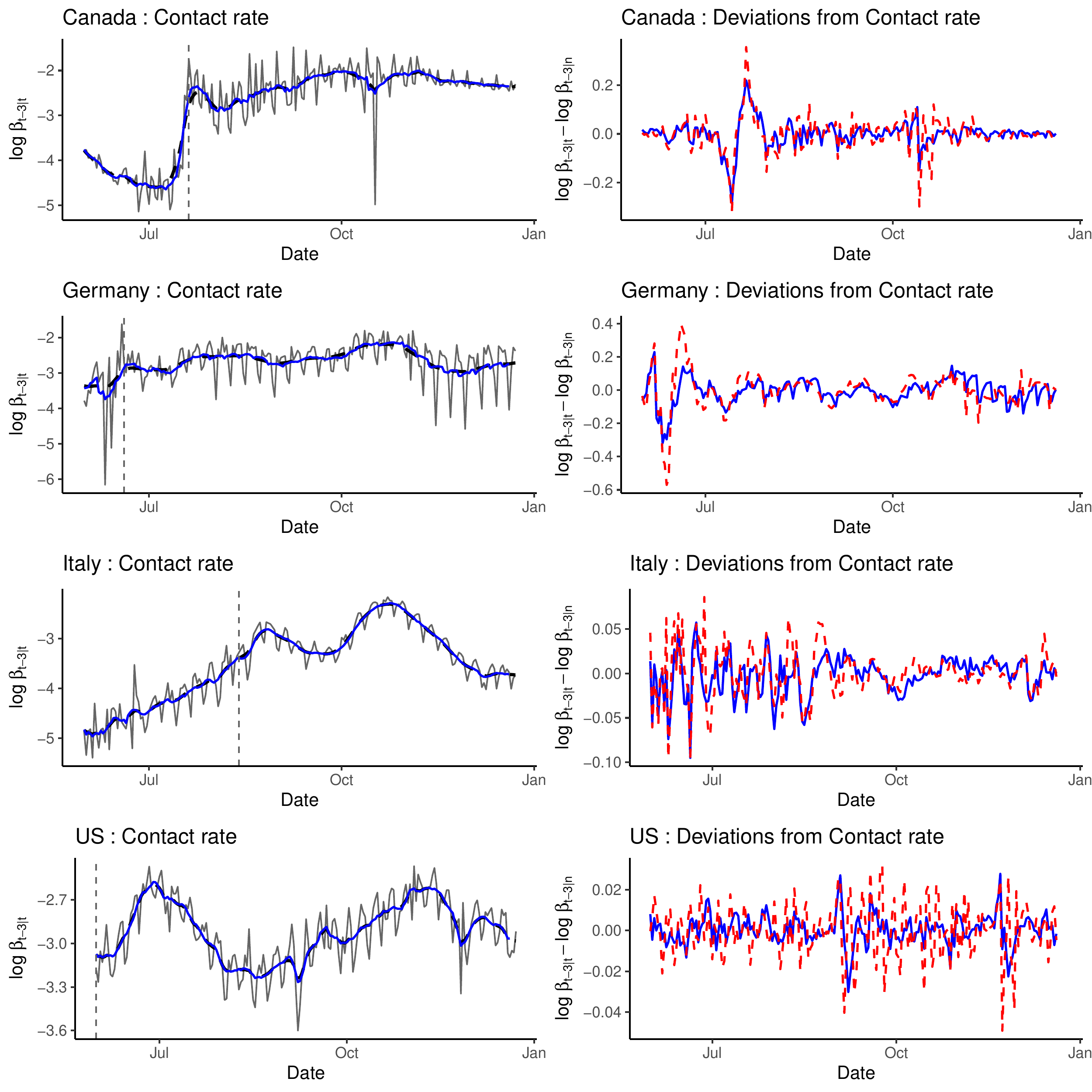}
	\caption[]{Real-time estimates for the log contact rate. The left panels display real-time contact rate estimates $\log \hat \beta_{t-3|t}$ (blue, solid), full sample estimates $\log \hat \beta_{t-3|n}$ (black, dashed), and the observable data $\log \tilde Y_{t-3}$ (gray) for Canada, Germany, Italy, and the United States. The dashed vertical line corresponds to the date where the real-time estimate for the reproduction rate exceeds $1.2$. The right panels show deviations from the full sample contact rate estimates for the real-time estimates $\log \hat \beta_{t-3|t} - \log \hat \beta_{t-3|n}$ (blue, solid) and the benchmark $\log \hat \beta_{t-3|t}^{benchmark} - \log \hat \beta_{t-3|n}$ (red, dashed).}
	\label{fig:6}
\end{figure}%

The results of the real-time experiment are visualized in figure \ref{fig:6}. The four panels on the left side sketch the resulting real-time estimates for the contact rate for the four countries under study, together with the (full sample) results of subsections \ref{subsec:1} and \ref{subsec:2}. As can be seen, the real-time estimates almost perfectly overlap with estimates using the full sample information. This implies that the real-time estimates are well suited for monitoring the current state of the contact rate. Minor deviations are visible for Canada at the end of July, and for Germany at the middle of June, while no greater deviations are visible for Italy and the US. The real-time estimates exceed the threshold $\hat{\mathcal{R}}_{t-3} = \hat{\beta}_{t-3|t} / \hat \gamma > 1.2$ on the same day as estimates based on the full sample for Italy (August 13) and for the US (May 31), where May 31 is the first observation for the real-time estimates. For Canada, the reproduction rate based on the full sample exceeds the threshold on July 21, one day after the real-time estimate. For Germany, the real-time estimates exceed the threshold for the reproduction rate on June 19,
two days before those based on the full sample. These findings again substantiate the usefulness of the fractional filter as a surveillance measure to monitor the current state of the pandemic. From the results in figure \ref{fig:6}, it follows that the primitive policy rule would have required governments to take action during the summer where case numbers were comparably low across the four countries under study. Such an early intervention would have likely reduced the economic and social costs compared to the lockdown measures as implemented at the end of year 2020.

The four panels on the right side of figure \ref{fig:6} display the deviations of $\log \hat \beta_{t-3|t}$ and $\log \hat \beta_{t-3|t}^{benchmark}$ from the log contact rate estimates based on the full sample information $\log \hat \beta_{t-3|n}$. Thus, they shed light on whether the fractional filter improves estimates for the contact rate compared to a rolling seven-day average that uses three forward-looking observations. For Italy and the US, the advantages of the fractional filter directly become apparent, as the benchmark exhibits greater deviations. For Canada and Germany, the fractional filter performs comparably well when large outliers occur, e.g.\ around July 20 for Canada and around June 20 for Germany. 
\section{Conclusion}\label{Conclusion}
To extract a time-varying signal for the COVID-19 contact rate from daily data on confirmed, recovered, and deceased cases, this paper introduces a novel unobserved components model. It models the log contact rate as a fractionally integrated process of unknown integration order. 
A computationally simple modification of the Kalman filter is introduced and is termed the fractional filter. It provides a closed-form expression for the prediction error that allows to estimate the model parameters by a 
conditional-sum-of-squares (CSS) estimator. 
The asymptotic theory for the CSS estimator is provided.
For the countries under study, estimation results are well in line with the chronology of the pandemic. They allow to draw inference on the impact of policy measures such as contact restrictions.
The new filtering method bears great potential as a monitoring device for the current state of the pandemic, as it yields reliable contact rate estimates at the current frontier of the data. 

As vaccines become more and more available, future research can generalize the model to include the number of vaccinated. For instance, this can be done by decomposing $1 = S_t + I_t + R_t + D_t + V_t$, where $V_t$ is the fraction of vaccinated. The states $R_t$ and $V_t$ should be non-overlapping as long as vaccines are not rolled out to recovered subjects. While vaccine recommendations vary over the different countries, some assign a lower priority to recovered subjects, so that $R_t$ and $V_t$ are non-overlapping at the early stage of the vaccine roll-out. Furthermore, mutations of the Coronavirus can be taken into account e.g.\ by allowing for a smooth transition between a contact rate with a low probability of virus transmission and one with a high probability. 

For applications beyond COVID-19 related data, the fractional filter offers a robust, flexible, and data-driven way for signal extraction of data of unknown persistence. 
It requires no prior assumptions on the integration order of a process, and thus provides a solution to model specification in the unobserved components literature. Due to its computational advantages compared to the classic Kalman filter, it allows to estimate unobserved components models with richer dynamics.

\section*{Acknowledgments}
The author thanks Nicolas Apfel, Uwe Hassler, Timon Hellwagner, Roland Jucknewitz, Alina Prechtl, Veronika P\"uschel, Lars Schlereth, Rolf Tschernig, Enzo Weber, and the participants of the Department Seminar at the University of Regensburg for very helpful comments.
\clearpage
\appendix
\setcounter{table}{0}
\numberwithin{table}{section}
\setcounter{equation}{0}
\numberwithin{equation}{section}
\section{Estimation results}
\begin{table}[ht]
	\centering
	\begin{tabular}{lrrrr}
		\hline
		& Canada & Germany & Italy & United States \\ 
		\hline
		$d$ & $1.2166$ & $1.2693$ & $1.4304$ & $1.2499$ \\ 
		 & $(0.3271)$ & $(0.1989)$ & $(0.1242)$ & $(0.2638)$ \\ 
		$\sigma_\eta^2$ & $0.0133$ & $0.0107$ & $0.0149$ & $0.0117$ \\ 
		 & $(0.1105)$ & $(0.1188)$ & $(0.2472)$ & $(0.2190)$ \\ 
		$\sigma_u^2$ & $0.2018$ & $0.7991$ & $0.3067$ & $0.0764$ \\ 
		& $(0.6128)$ & $(0.7636)$ & $(0.5920)$ & $(0.3791)$ \\ 
		\hline
	\end{tabular}
\caption{Estimation results $\hat \vtheta$ from the CSS estimator as described in section \ref{Ch:est} for Canada, Germany, Italy, and the United States. Standard errors are denoted in parentheses and were calculated based on the inverse of the numeric Hessian matrix, see theorem \ref{th:norm}.}
\label{ta:1}
\end{table}
\section{Monte Carlo evidence}\label{MC}
The finite sample performance of the CSS estimator is assessed in a Monte Carlo study, where, to be in line with \eqref{eq:mod2}, the data-generating mechanism is given by 
\begin{align}\label{eq:sim}
	y_t = x_t + u_t, \qquad \Delta_+^{d_0} x_t = \eta_t, \qquad t=1,...,n.
\end{align}
$u_t \sim NID(0, \sigma_{u,0}^2)$, $\eta_t \sim NID(0,  \sigma_{\eta, 0}^2)$, $u_t$, $\eta_t$ are uncorrelated, and $\sigma_{\eta, 0}^2 = \rho \sigma_{u, 0}^2$ so that $\rho$ controls the signal-to-noise ratio. The integration orders $d_0 \in \{0.75, 1.25, 1.75\}$ cover the relevant interval for the applications in section \ref{application}, while $\rho \in \{0.5, 1, 2\}$ captures high and low signal-to-noise ratios. The variance parameter is set to $\sigma_{u, 0}^2 = 1$. Different sample sizes $n \in \{100, 200, 300\}$ covering the relevant regions for the applications in section \ref{application} are considered.
The parameters $\theta_0 = (d_0, \sigma_{\eta, 0}^2, \sigma_{u, 0}^2)$ are estimated via the CSS estimator as described in section \ref{Ch:est}. For each specification, $1000$ replications are simulated, and starting values are set to $\theta_{start} = (1, 1, 1)$.

\begin{table}[p]
	\centering
	\begin{tabular}{llrrrrrrrrr}
		\hline
		$\rho$ & $d_0$ & $\hat d$ & $\hat d_{EW}^{.45}$ & $\hat d_{EW}^{.50}$ & $\hat d_{EW}^{.55}$ & $\hat d_{EW}^{.60}$ & $\hat d_{EW}^{.65}$ & $\hat d_{EW}^{.70}$ & $MSE_x$ & $R_x^2$ \\ 
		\hline
		\multicolumn{11}{c}{$n=100$}\\
		\hline
		.5 & 0.75 & 0.0641 & 0.1021 & 0.0804 & 0.0762 & 0.0736 & 0.0728 & 0.0775 & 0.4786 & 0.6747 \\ 
		.5 & 1.25 & 0.0387 & 0.1011 & 0.0721 & 0.0664 & 0.0694 & 0.0789 & 0.1054 & 0.3719 & 0.9796 \\ 
		.5 & 1.75 & 0.0285 & 0.0876 & 0.0620 & 0.0576 & 0.0637 & 0.0809 & 0.1293 & 0.3418 & 0.9992 \\ 
		1 & 0.75 & 0.0409 & 0.0943 & 0.0673 & 0.0585 & 0.0505 & 0.0446 & 0.0433 & 0.6245 & 0.7914 \\ 
		1 & 1.25 & 0.0299 & 0.0978 & 0.0644 & 0.0535 & 0.0484 & 0.0465 & 0.0570 & 0.4880 & 0.9867 \\ 
		1 & 1.75 & 0.0239 & 0.0851 & 0.0539 & 0.0470 & 0.0453 & 0.0475 & 0.0710 & 0.4258 & 0.9995 \\ 
		2 & 0.75 & 0.0277 & 0.0919 & 0.0615 & 0.0504 & 0.0407 & 0.0318 & 0.0264 & 0.7861 & 0.8711 \\ 
		2 & 1.25 & 0.0231 & 0.0977 & 0.0601 & 0.0489 & 0.0393 & 0.0323 & 0.0319 & 0.6282 & 0.9915 \\ 
		2 & 1.75 & 0.0204 & 0.0830 & 0.0511 & 0.0422 & 0.0372 & 0.0325 & 0.0384 & 0.5306 & 0.9997 \\ 
		\hline
		\multicolumn{11}{c}{$n=200$}\\
		\hline
		 .5 & 0.75 & 0.0232 & 0.0662 & 0.0446 & 0.0396 & 0.0393 & 0.0427 & 0.0488 & 0.3985 & 0.8158 \\ 
		 .5 & 1.25 & 0.0154 & 0.0615 & 0.0394 & 0.0320 & 0.0313 & 0.0381 & 0.0541 & 0.3432 & 0.9934 \\ 
		 .5 & 1.75 & 0.0128 & 0.0519 & 0.0348 & 0.0281 & 0.0260 & 0.0328 & 0.0561 & 0.3287 & 0.9999 \\ 
		 1 & 0.75 & 0.0171 & 0.0641 & 0.0390 & 0.0307 & 0.0256 & 0.0238 & 0.0248 & 0.5479 & 0.8742 \\ 
		 1 & 1.25 & 0.0124 & 0.0614 & 0.0378 & 0.0280 & 0.0225 & 0.0214 & 0.0267 & 0.4502 & 0.9956 \\ 
		 1 & 1.75 & 0.0106 & 0.0513 & 0.0335 & 0.0260 & 0.0206 & 0.0192 & 0.0276 & 0.4085 & 0.9999 \\ 
		 2 & 0.75 & 0.0128 & 0.0622 & 0.0372 & 0.0275 & 0.0202 & 0.0159 & 0.0140 & 0.7169 & 0.9185 \\ 
		 2 & 1.25 & 0.0104 & 0.0620 & 0.0370 & 0.0268 & 0.0193 & 0.0149 & 0.0145 & 0.5815 & 0.9972 \\ 
		 2 & 1.75 & 0.0091 & 0.0510 & 0.0331 & 0.0252 & 0.0185 & 0.0141 & 0.0148 & 0.5068 & 0.9999 \\ 
		\hline
		\multicolumn{11}{c}{$n=300$}\\
		\hline
		.5 & 0.75 & 0.0157 & 0.0448 & 0.0339 & 0.0284 & 0.0274 & 0.0311 & 0.0386 & 0.3770 & 0.8594 \\ 
		.5 & 1.25 & 0.0106 & 0.0404 & 0.0288 & 0.0218 & 0.0199 & 0.0242 & 0.0386 & 0.3394 & 0.9964 \\ 
		.5 & 1.75 & 0.0089 & 0.0361 & 0.0265 & 0.0195 & 0.0164 & 0.0189 & 0.0364 & 0.3278 & 0.9999 \\ 
		1 & 0.75 & 0.0120 & 0.0423 & 0.0301 & 0.0224 & 0.0185 & 0.0173 & 0.0190 & 0.5227 & 0.9031 \\ 
		1 & 1.25 & 0.0086 & 0.0402 & 0.0282 & 0.0199 & 0.0157 & 0.0145 & 0.0187 & 0.4442 & 0.9976 \\ 
		1 & 1.75 & 0.0075 & 0.0361 & 0.0258 & 0.0187 & 0.0145 & 0.0125 & 0.0177 & 0.4069 & 1.0000 \\ 
		2 & 0.75 & 0.0093 & 0.0412 & 0.0288 & 0.0203 & 0.0153 & 0.0118 & 0.0105 & 0.6895 & 0.9366 \\ 
		2 & 1.25 & 0.0072 & 0.0401 & 0.0276 & 0.0193 & 0.0144 & 0.0109 & 0.0103 & 0.5712 & 0.9985 \\ 
		2 & 1.75 & 0.0064 & 0.0360 & 0.0257 & 0.0186 & 0.0140 & 0.0102 & 0.0099 & 0.5037 & 1.0000 \\ 
		\hline
	\end{tabular}
	\caption{Mean squared error (MSE) and $R^2_x$ for $d_0$ and $x_t$ in \eqref{eq:sim}. The columns $\hat d$ and $\hat d_{EW}^j$ show the MSE for the CSS estimator of $d_0$ as well as for the exact local Whittle estimator of \cite{Shi2010} for $m = \lfloor n^{j} \rfloor$ Fourier frequencies, $j \in \{0.45, 0.50, 0.55, 0.60, 0.65, 0.70\}$. $MSE_x$ displays the mean squared error for $x_t$, while $R^2_x$ is the coefficient of determination, see \eqref{eq:r2}.}
	\label{ta:2}
\end{table}

In addition to the CSS estimates, estimation results for $d_0$ from the exact local Whittle estimator of \cite{Shi2010} are reported as benchmarks for $m = \lfloor n^{j} \rfloor$ Fourier frequencies, $j \in \{0.45, 0.50, 0.55, 0.60, 0.65, 0.70\}$. Finally, the mean squared error $MSE_x$ and the coefficient of determination $R_x^2$ for the estimation of $x_t$, that are calculated via 
\begin{align}\label{eq:r2}
MSE_x = \frac{1}{n} \sum_{t=1}^{n}(x_t - {x}_{t|n}(\hat \vtheta))^2, \qquad R_x^2 = 1 - \frac{\sum_{t=1}^{n}(x_t - {x}_{t|n}(\hat \vtheta))}{\sum_{t=1}^{n}(x_t - \bar{x})^2},
\end{align}
are reported, and indicate how well $x_t$ is estimated by the fractional filter \eqref{KS:1}. 

The results for the Monte Carlo study are contained in table \ref{ta:2}. Not surprisingly, the parametric CSS estimator outperforms the exact local Whittle estimator. However, gains are quite large in terms of the MSE for the integration order for all $n \in \{100, 200, 300\}$ and all combinations of $\rho$ and $d_0$. The mean squared error of the integration order estimate becomes smaller for higher $d_0$, which is plausible as the fraction of total variation of $y_t$ generated by $x_t$ increases with $d_0$. 
For the same reason, it decreases with increasing $\rho$. The same conclusions on the precision with which $d_0$ is estimated hold for the mean squared error of $x_t$, which decreases as $n$, $d$, and $\rho$ increase. 
The proportion of explained variation of $x_t$, measured by $R^2_x$, is high and thus $x_t$ is estimated well via \eqref{KS:1}. Particularly for $d=1.25$, which is the relevant case for the applications in section \ref{application}, the $R^2_x$ is close to unity for all $n$. 
\section{Proofs}\label{app_proofs}
\begin{proof}[Proof of Lemma \ref{L:Kalman}]
	First, note that $\E_\vtheta(x_{t+1}|\mathcal{F}_t) = \E_\vtheta(y_{t+1}|\mathcal{F}_t)$, so that it is sufficient to derive the latter expression. 
	For this, consider the reduced form of \eqref{eq:mod2}, which follows from taking fractional differences and utilizing the aggregation properties of MA processes, see \cite{GraMor1976}, so that 
	\begin{align}\label{eq:a1}
	\Delta_+^d y_t = \eta_t + \Delta_+^d u_t = \eta_t + \sum_{i=0}^{t-1}\pi_i(d)u_{t-i} = \sum_{i=0}^{t-1} \phi_i(\vtheta) \varepsilon_{t-i} =  \phi(L, \vtheta)\varepsilon_t,
	\end{align}
	with $\phi_0(\vtheta) = 1$, $\varepsilon_t \sim WN(0, \sigma_\varepsilon^2)$, and $\phi(L, \vtheta)$ is invertible. $\sigma_\varepsilon^2$ and the coefficients in $\phi(L, \vtheta)$ can be derived by matching the autocovariance functions of \eqref{eq:a1}, see \citet[eqn.\ 2.6]{Wat1986}, and depend non-linearly on $\vtheta$. However, they are not required for the proof.
	Solving for $\varepsilon_{t}$ yields
	\begin{align*}
		\varepsilon_{t} = \phi(L, \vtheta)^{-1} \Delta_+^d y_t = y_t - \sum_{i=1}^{\infty} A_i(\vtheta) y_{t-i}.
	\end{align*}
	From the type II definition of fractional integration, see assumption \ref{asu:2}, it follows that $\Cov_\vtheta (y_t, y_j) = 0$ for all $j \leq 0$, $t > 0$, and thus $\E_\vtheta(y_{t+1}|\mathcal{F}_t) = \sum_{i=1}^{t}A_i(\vtheta)y_{t+1-i}$. The (yet unknown) coefficients $A_i(\vtheta)$ follow from the Yule-Walker equations 
	\begin{align*}
	\bvec  \Cov_\vtheta(y_{t+1}, y_{t}) \\ \Cov_\vtheta(y_{t+1}, y_{t-1}) \\ \vdots \\ \Cov_\vtheta(y_{t+1}, y_{1}) \evec = 
	\bmat \Var_\vtheta(y_t) & \Cov_\vtheta(y_{t-1}, y_t) & \cdots & \Cov_\vtheta(y_{1}, y_t) \\
				\Cov_\vtheta(y_{t}, y_{t-1}) & \Var_\vtheta(y_{t-1}) & \cdots & \Cov_\vtheta(y_{1}, y_{t-1}) \\
				\vdots & \vdots & \ddots & \vdots \\
				\Cov_\vtheta(y_{t}, y_{1}) & \Cov_\vtheta(y_{t-1}, y_{1}) & \cdots & \Var_\vtheta(y_{1})  \emat
				\bvec A_1(\vtheta) \\ A_2(\vtheta) \\ \vdots \\ A_t(\vtheta) \evec,
	\end{align*}
	so that by defining the vectors $\mA(\vtheta) = (A_1(\vtheta), ..., A_t(\vtheta))$, $y_{t:1} = (y_t,...,y_1)'$,
	and solving the Yule-Walker equations for $\mA(\vtheta)$, one has $\mA(\vtheta) = \Cov_\vtheta(y_{t+1}, y_{t:1})\Var_\vtheta(y_{t:1})^{-1}$, which implies 
	\begin{align*}
		\E_{\vtheta}(y_{t+1} | \mathcal{F}_t) = \sum_{i=1}^{t} A_i(\vtheta)y_{t-i} = \Cov_\vtheta^{}(y_{t+1}, y_{t:1})\Var_\vtheta(y_{t:1})^{-1}y_{t:1}.
	\end{align*}
	From $\Cov_\vtheta^{}(y_{t+1}, y_{t:1}) = \Cov_\vtheta^{}(x_{t+1}, y_{t:1}) = \sum_{i=1}^{t} \pi_i(-d)\Cov_\vtheta^{}(\eta_{t+1-i}, y_{t:1})$, see assumption \ref{asu:2}, lemma \ref{L:Kalman} follows.
	

\end{proof}

\begin{proof}[Proof of Theorem \ref{th:cons}]
	First, the model in \eqref{eq:mod2} is shown
	to be identified. Identification follows if the parameters 
	$\sigma_\eta^2$, $\sigma_u^2$ can be recovered from the autocovariance function of the reduced form $\phi(L, \vtheta) \varepsilon_t$ in \eqref{eq:a1}. To see this, consider the covariances $\Var_\vtheta(\phi(L, \vtheta)\varepsilon_{t}) = \sigma_\varepsilon^2 \sum_{i=0}^{t-1}\phi_i^2(\vtheta) = \sigma_\eta^2 + \sigma_u^2 \sum_{i=0}^{t-1}\pi_i^2(d)$, and $\Cov_\vtheta(\phi(L, \vtheta)\varepsilon_{t}, \phi(L, \vtheta)\varepsilon_{t-1}) = \sigma_\varepsilon^2 \sum_{i=0}^{t-2}\phi_i(\vtheta)\phi_{i+1}(\vtheta) =  \sigma_u^2 \sum_{i=0}^{t-2}\pi_i(d)\pi_{i+1}(d)$. In matrix form this gives
	\begin{align}\label{eq:23}
		\sigma_\varepsilon^2 \bvec \sum_{i=0}^{t-1}\phi_i^2(\vtheta) \\ \sum_{i=0}^{t-2}\phi_i(\vtheta)\phi_{i+1}(\vtheta)  \evec &= \bmat 1 & \sum_{i=0}^{t-1}\pi_i^2(d) \\ 0 & \sum_{i=0}^{t-2}\pi_i(d)\pi_{i+1}(d) \emat \bvec \sigma_\eta^2 \\ \sigma_u^2 \evec,
		\end{align}
		so that solving for $(\sigma_\eta^2, \sigma_u^2)'$ yields
		\begin{align*}
		\bvec \sigma_\eta^2 \\ \sigma_u^2 \evec &= \frac{1}{\sum_{i=0}^{t-2}\pi_i(d)\pi_{i+1}(d)} \bmat
		\sum_{i=0}^{t-2}\pi_i(d)\pi_{i+1}(d) & - \sum_{i=0}^{t-1}\pi_i^2(d) \\
		0 & 1 \emat	 \bvec \sum_{i=0}^{t-1}\phi_i^2(\vtheta) \\ \sum_{i=0}^{t-2}\phi_i(\vtheta)\phi_{i+1}(\vtheta)  \evec \sigma_\varepsilon^2,
	\end{align*}
	and thus $(\sigma_\eta^2, \sigma_u^2)'$ can be uniquely recovered from the reduced form. 
	The assumption that $d > 0$ is crucial, as it guarantees $\sum_{i=0}^{t-2}\pi_i(d)\pi_{i+1}(d) \neq 0$, so that the matrix in \eqref{eq:23} has full rank.

	Next, the CSS estimator based on the reduced form \eqref{eq:a1} is derived and is shown to be identical to \eqref{eq:QML}.
	Multiplying \eqref{eq:a1} by $\phi(L, \vtheta)^{-1}$ yields $\varepsilon_t = \phi(L, \vtheta)^{-1} \Delta_+^d y_t$, based on which a reduced form CSS estimator can be constructed. Define $\psi_+(L, \vtheta) = [\phi(L, \vtheta)^{-1}]_+ = [1 - \sum_{i=1}^\infty \psi_i(\vtheta) L^i]_+$, as the (truncated) inverse of $\phi(L, \vtheta)$, and denote
	 $\varepsilon_t(\vtheta) = [\phi(L, \vtheta)^{-1}]_+\Delta_+^d y_t = \psi_+(L, \vtheta)\Delta_+^d y_t$ as the reduced form residual given the observable variables $y_1,...,y_n$ and $\vtheta$. From $\varepsilon_t(\vtheta)$, the reduced form CSS estimator is 
	\begin{align}\label{eq:QML2}
		\hat \vtheta = \arg \min_{\vtheta \in \mTheta}
		R(\vtheta), \qquad R(\vtheta)= \frac{1}{n}\sum_{t=1}^n\varepsilon_t^2(\vtheta),
	\end{align}
	and equals the CSS estimator in \eqref{eq:QML}. To see this, add and subtract $y_t$ from $\varepsilon_t(\vtheta) = \psi_+(L, \vtheta) \Delta_+^d y_t$, so that $y_t = (1-\psi_+(L, \vtheta) \Delta_+^d) y_t + \varepsilon_t(\vtheta)$, and plug $y_t$ into the conditional expectation in \eqref{eqn:v}
	\begin{align*}
		v_t(\vtheta) &= y_t - \E_\vtheta(y_t | \mathcal{F}_{t-1}) = y_t - \E_\vtheta\left[(1-\psi_+(L, \vtheta) \Delta_+^d) y_t| \mathcal{F}_{t-1}\right] = \psi_+(L, \vtheta)\Delta_+^d y_t = \varepsilon_t(\vtheta).
	\end{align*}
	The third equality follows from $(1-\psi_+(L, \vtheta) \Delta_+^d) y_t$ being $\mathcal{F}_{t-1}$-measurable, since $\psi(L, \vtheta) = 1 - \sum_{i=1}^{\infty}\psi_i(\vtheta) L^i$ and $\pi_0(d)=1$. Thus, the contemporaneous $y_t$ cancel in the expectation operator and the whole term can be taken out of the expectation operator. From $v_t(\vtheta)=\varepsilon_t(\vtheta)$ it follows that the optimization problems in \eqref{eq:QML} and \eqref{eq:QML2} are identical. 
	
	Next, the integration order of the residuals is assessed. Since $y_t \sim I(d_0)$, the residuals satisfy
	\begin{align}\label{eq:resid}
		\varepsilon_t(\vtheta) = \psi_+(L, \vtheta)\Delta_+^d y_t = \psi_+(L, \vtheta)\Delta_+^{d - d_0}\eta_t + \psi_+(L, \vtheta)\Delta_+^{d}u_t \sim I(d_0 - d).
	\end{align}
	For $d_0 - d < 1/2$ the residuals are stationary, while for $d_0 - d > 1/2$ they are nonstationary. As the asymptotic behavior of the objective function changes around $d_0 - d = 1/2$, the objective function does not uniformly converge on the set of admissible values for $d$. 
	The same problem is addressed by \cite{HuaRob2011} and by \cite{Nie2015} for ARFIMA models encompassing \eqref{eq:a1}. 
	\citet[eqn.\ 8]{Nie2015} shows that a weak law of large numbers (WLLN) applies to the sum of squared residuals whenever $d_0 - d < 1/2$, while the sum of squared residuals diverges in probability whenever $d_0 - d \geq 1/2$, which translates into 
	\begin{align}\label{plim}
		\plim_{n \to \infty} \frac{1}{n} \sum_{t=1}^n \varepsilon_t^2(\vtheta) = \begin{cases}
			\E[\tilde \varepsilon_t^2(\vtheta)] & \text{if } d_0 - d < 1/2, \\
			\infty & \text{else.}
		\end{cases}
	\end{align}
$\tilde \varepsilon_t(\vtheta) = \psi(L, \vtheta) \phi(L, \vtheta_0)\Delta^{d-d_0}\varepsilon_t$ is the untruncated residual generated by the untruncated $\Delta^d$  and $\psi(L, \vtheta)$. In addition, letting $D^*(\kappa) = D \cap \{d: d_0 - d \leq 1/2 - \kappa\} $, $0 < \kappa < 1/2$ denote the region of the parameter space where $\varepsilon_t(\vtheta)$ is stationary, \citet[eqn.\ 13]{Nie2015} shows that for any constant $K>0$ there exists a fixed $\bar \kappa > 0$ such that
\begin{align}\label{eqn:a2}
\Pr \left( \inf_{d \in D \setminus D^*(\bar \kappa) \cap \theta \in \mTheta} \frac{1}{n}\sum_{t=1}^n \varepsilon_t^2(\vtheta) > K \right) \to 1, \quad \text{as } n \to \infty,
\end{align}
implying that $\Pr(\hat d \in D^*(\bar \kappa) \cap \theta \in \mTheta) \to 1$ as $n \to \infty$. From \eqref{eqn:a2} it follows  that the relevant parameter space asymptotically reduces to the stationary region $\mTheta^*(\bar \kappa) = \{\vtheta | \vtheta \in \mTheta, d \in D^*(\bar \kappa)\}$.

Since the model is identified, for consistency it remains to be shown that a uniform weak law of large numbers (UWLLN) holds for the objective function within the stationary region of the parameter space. A UWLLN holds if both, the objective function  
and the supremum of the gradient, satisfy a WLLN, see \citet[thm.\ 4.2 and eqn. 4.4]{Woo1994} and \citet[cor.\ 2.2]{New1991}. While a WLLN for the objective function follows directly from \eqref{plim}, it remains to be shown that
\begin{align}\label{eq:UWLLN}
	\sup_{\vtheta \in \mTheta^*(\kappa)} \left\lvert \frac{\partial R(\vtheta)}{\partial \vtheta} \right \rvert = O_p(1),
\end{align}
for any fixed $0 < \kappa < 1/2$.

To prove \eqref{eq:UWLLN}, it will be helpful to note that for a white noise process $\varepsilon_t$, MA weights $\sum_{i=0}^\infty |m_{h, i}(\vtheta)| < \infty$, $h=1,2$, and the set $\tilde \mTheta = \{\vtheta | \vtheta \in \mTheta,  d_0 - d < 1/2\}$, it holds that
\begin{align}\label{plim:fracdiff}
\sup_{\vtheta \in \tilde \mTheta} \left\lvert \frac{1}{n} \sum_{t=1}^n \left[
\frac{\partial^j \Delta_+^{d-d_0}}{\partial d^j} \sum_{i=0}^\infty m_{1,i}(\vtheta)\varepsilon_{t-i}
\right]
\left[
\frac{\partial^k \Delta_+^{d-d_0}}{\partial d^k} \sum_{i=0}^\infty m_{2,i}(\vtheta)\varepsilon_{t-i}
\right]
\right\rvert = O_p(1),
\end{align}
for $j, k \geq 0$ as shown by \citet[lemma B.3]{Nie2015}. 

Now, consider the partial derivatives of \eqref{eq:QML2}
\begin{align}\label{eq:a3}
	\frac{\partial R(\vtheta)}{\partial \vtheta} &= \frac{2}{n}\sum_{t=1}^{n} \varepsilon_{t}(\vtheta) \frac{\partial \varepsilon_{t}(\vtheta)}{\partial \vtheta}, \qquad \frac{\partial \varepsilon_{t}(\vtheta)}{\partial \vtheta} = \frac{\partial \psi_+(L, \vtheta)}{\partial \vtheta} \Delta_+^d y_t + \psi_+(L, \vtheta) \frac{\partial \Delta_+^{d-d_0}}{\partial \vtheta}\Delta_+^{d_0}y_t.
\end{align}
Since $\psi_+(L, \vtheta)$ satisfies the absolute summability condition for \eqref{plim:fracdiff}, it follows that
\begin{align}\label{eq:UWLLN3}
	\sup_{\vtheta \in \mTheta^*(\kappa)} \left\lvert \frac{1}{n} \sum_{t=1}^n
	\varepsilon_t(\vtheta) \psi_+(L, \vtheta) \frac{\partial \Delta_+^{d-d_0}}{\partial d}\Delta_+^{d_0}y_t \right\rvert = O_p(1),
\end{align}
while the partial derivatives of $\Delta_+^{d-d_0}$ w.r.t.\ $\sigma_\eta^2, \sigma_u^2$ are zero.

For the remaining term in \eqref{eq:a3}, note that the sum of absolute coefficients of the truncated polynomial $\psi_+(L, \vtheta)$ is bounded by the sum of absolute coefficients of the untruncated polynomial $\psi(L, \vtheta) = \phi(L, \vtheta)^{-1}$. Thus, it is sufficient to prove absolute summability of the coefficients in $\partial \psi(L, \vtheta)/\partial \vtheta = - \psi(L, \vtheta)^2 (\partial \phi(L, \vtheta)/\partial \vtheta)$. Absolute summability of the coefficients in $\partial \phi(L, \vtheta)/\partial \vtheta$ is shown in lemma \ref{LD} in appendix \ref{app:D}. Since $\psi(L, \vtheta)$ is stable,  $\partial \psi(L, \vtheta)/\partial \vtheta$ satisfies the absolute summability condition for \eqref{plim:fracdiff} and thus
\begin{align}\label{eq:UWLLN4}
\sup_{\vtheta \in \mTheta^*(\kappa)} \left\lvert \frac{1}{n} \sum_{t=1}^n
\varepsilon_t(\vtheta) \frac{\partial \psi_+(L, \vtheta)}{\partial \vtheta} \Delta_+^d y_t \right\rvert = O_p(1).
\end{align}
From \eqref{eq:UWLLN3} and \eqref{eq:UWLLN4} it follows that \eqref{eq:UWLLN} holds. Consequently, the supremum of the gradient satisfies a WLLN for $\vtheta \in \mTheta^*(\kappa)$, which generalizes the pointwise convergence of the objective function to weak convergence, implying that a UWLLN holds for the objective function. Since the model is identified, consistency of the CSS estimator follows from the UWLLN together with \eqref{eqn:a2}, and thus $\hat \vtheta \pto \vtheta_0$ as $n\to \infty$, see \citet[thm.\ 4.3]{Woo1994}. 
\end{proof}


\begin{proof}[Proof of Theorem \ref{th:norm}]
Since the CSS estimator is consistent, see theorem \ref{th:cons}, the asymptotic distribution theory can be inferred from a Taylor expansion of the score function about $\vtheta_0$
\begin{align}\label{eq:taylor}
	0 = \sqrt{n} \frac{\partial R(\vtheta)}{\partial \vtheta} \Bigg\rvert_{\vtheta = \hat \vtheta} = \sqrt{n} \frac{\partial R(\vtheta)}{\partial \vtheta} \Bigg\rvert_{\vtheta = \vtheta_0} + \sqrt{n} \frac{\partial^2 R(\vtheta)}{\partial \vtheta \partial \vtheta'}\Bigg\rvert_{\vtheta = \bar \vtheta} (\hat \vtheta - \vtheta_0),
\end{align}
where the entries in $\bar \vtheta$ satisfy $|\bar \vtheta_i  - \vtheta_{0, i}| \leq | \hat \vtheta_{i} - \vtheta_{0, i} |$ for all $i=1,2,3$, and $\vtheta_i$ denotes the $i$-th entry of $\vtheta = (d, \sigma_\eta^2, \sigma_u^2)'$, $i=1,2,3$. The score function at $\vtheta_0$ follows from \eqref{eq:a3} 
\begin{align}\label{eqn:Sn1}
\sqrt{n} \frac{\partial R(\vtheta)}{\partial \vtheta} \Bigg\rvert_{\vtheta = \vtheta_0} &=\frac{2}{\sqrt{n}}\sum_{t=1}^n 
	\varepsilon_t(\vtheta_0) \frac{\partial \varepsilon_t(\vtheta)}{\partial \vtheta}\Bigg\rvert_{\vtheta = \vtheta_0}
 = S_n + o_p(1), 
\intertext{where}
S_n &= \frac{2}{\sqrt{n}}\sum_{t=1}^n 
\varepsilon_t \frac{\partial \tilde\varepsilon_t(\vtheta)}{\partial \vtheta}\Bigg\rvert_{\vtheta = \vtheta_0} 
. \label{eq:Sn2}
\end{align}
$\tilde \varepsilon_t(\vtheta) = \psi(L, \vtheta)\phi(L, \vtheta_0)\Delta^{d-d_0} \varepsilon_t$ is the untruncated residual generated by the untruncated $\Delta^d$ and $\psi(L, \vtheta) = 1-\sum_{i=1}^\infty \psi_i(\vtheta) L^i$, and the second equality in \eqref{eqn:Sn1} is shown to hold by \citet[pp.\ 135-136]{Rob2006}. 
In the following, let $S_n^{(j)}$ denote the $j$-th entry of $S_n$ holding the partial derivative w.r.t.\ $\vtheta_j$, $j=1,2,3$, and let $C_{1, j}(L, \vtheta) = \sum_{i=1}^\infty C_{1, j, i}(\vtheta)L^i =  \phi(L, \vtheta_0)(\partial/\partial \vtheta_j) [\psi(L, \vtheta)\Delta^{d-d_0}]$ denote the coefficients of the partial derivative of $\tilde \varepsilon_{t}(\vtheta)$ w.r.t.\ $\vtheta_j$. 

To derive the asymptotic distribution theory for the CSS estimator, a central limit theorem (CLT) is shown to hold for the score function at $\vtheta_0$. Next, it is proven that a UWLLN holds for the Hessian matrix by showing that the Hessian matrix and its first partial derivatives satisfy a WLLN \citep[thm.\ 4.2]{Woo1994}. The UWLLN allows to evaluate the Hessian matrix in \eqref{eq:taylor} at $\vtheta_0$ and yields the asymptotic distribution of $\sqrt{n}(\hat \vtheta - \vtheta_0)$. As the reduced form coefficients $\phi(L)$ depend non-trivially on $\vtheta$, no analytical expression for the asymptotic variance of the CSS estimator is provided. Instead, it will be shown that the CSS estimator is asymptotically normally distributed, and its asymptotic variance is shown to exist. This allows to estimate $\Var(\hat \vtheta)$ e.g.\ via the inverse of the numerical Hessian matrix. 

Starting with the score function, similar to \citet[p.\ 175]{Nie2015} a CLT can be inferred from the Cram\'er-Wold device by showing that for any $3$-dimensional vector $\mu = (\mu_1, \mu_2, \mu_3)'$, it holds that $\mu' S_n = \sum_{j=1}^3 \mu_j S_n^{(j)} \dto \mathrm{N}(0, \mu' \mOmega_0 \mu)$. To see this, define the $\sigma$-algebra $\tilde{ \mathcal{F}}_t = \sigma(\{\varepsilon_s, s\leq t\})$ generated by the white noise $\varepsilon_t$ and its lags. Next, note that in \eqref{eq:Sn2} the term $\varepsilon_t [\partial \tilde \varepsilon_t(\vtheta)/ \partial \vtheta \big\rvert_{\vtheta = \vtheta_0}]$ adapted to $\tilde{ \mathcal{F}}_t$ is a stationary MDS, since $\varepsilon_t$ is white noise, the partial derivatives are $\tilde{\mathcal{F}}_{t-1}$-measurable, and the coefficients of the partial derivatives are absolutely summable, as shown in the proof of theorem \ref{th:cons}. It follows for $\mu'S_n = 2n^{-1/2} \sum_{t=1}^n \nu_t$ with
\begin{align*}
	\nu_t = \sum_{j=1}^3 \nu_{j,t}  \qquad \nu_{j,t} = \mu_j\varepsilon_t \frac{\partial \tilde \varepsilon_t(\vtheta)}{\partial \vtheta_j}\Bigg\rvert_{\vtheta = \vtheta_0},
\end{align*}
that $\nu_t$ adapted to $\tilde{\mathcal{F}_t}$ is a stationary MDS. Similar to \citet[p.\ 175]{Nie2015}, by the law of large numbers for stationary and ergodic processes, the sum of conditional variances for $\mu' S_n$ with $S_n$ as given in \eqref{eq:Sn2} is then 
\begin{align}
\frac{1}{n}\sum_{t=1}^{n}\E(\nu_t^2| \tilde{ \mathcal{F}}_{t-1}) &= \frac{1}{n}\sum_{t=1}^{n} \sum_{j, k=1}^{3} \E \left[
	\nu_{j,t}\nu_{k,t}  | \tilde{ \mathcal{F}}_{t-1}
\right] =  \sum_{j,k=1}^{3} \mu_j \mu_k \sigma_{\varepsilon, 0}^2 \frac{1}{n}\sum_{t=1}^{n} \frac{\partial \tilde \varepsilon_t(\vtheta)}{\partial \vtheta_j}\Bigg\rvert_{\theta = \theta_0}\frac{\partial \tilde \varepsilon_t(\vtheta)}{\partial \vtheta_k}\Bigg\rvert_{\theta = \theta_0} \nonumber \\
&\pto  \sum_{j,k=1}^{3} \mu_j \mu_k \sigma_{\varepsilon, 0}^4 \sum_{i=1}^{\infty}C_{1, j, i}(\vtheta_0) C_{1, k, i}(\vtheta_0) = \sum_{j,k=1}^{3} \mu_j \mu_k \Omega_0^{(j, k)}.
\end{align}
In $C_{1,j}(\vtheta_0) = \phi(L, \vtheta_0)(\partial/\partial \vtheta_j) [\psi(L, \vtheta)\Delta^{d-d_0}]\big\rvert_{\theta = \theta_0}$, the partial derivatives of the first polynomial $\partial \psi(L, \vtheta)/\partial \vtheta_j = - 2 \psi(L, \vtheta)(\partial \phi(L, \vtheta)/\partial \vtheta_j)$ are absolutely summable for all $j=1,2,3$, as $\psi(L, \vtheta)$ and $\partial \phi(L, \vtheta)/\partial \vtheta_j$ are absolutely summable, see lemma \ref{LD} in appendix \ref{app:D}. Furthermore, $(\partial / \partial d) \Delta^{d-d_0}\big\rvert_{\vtheta = \vtheta_0} = \sum_{j=1}^\infty j^{-1} L^j$ \citep[p.\ 175]{Nie2015}, so that $\sum_{i=1}^{\infty}C_{1, j, i}(\vtheta_0) C_{1, k, i}(\vtheta_0) = O(1)$. Consequently, by the CLT for stationary MDS \citep[see e.g.][thm.\ 6.2.3]{Dav2000}  $S_n \dto \mathrm{N}(0, 4 \mOmega_0)$.

To evaluate the Hessian matrix in \eqref{eq:taylor} at $\vtheta_0$, it remains to be shown that a UWLLN applies to the Hessian matrix \citep[thm.\ 4.4]{Woo1994}, for which it is sufficient to show that a WLLN holds for the Hessian matrix and for the supremum of its first partial derivatives
\begin{align}\label{Hessian:WLLN}
\sup_{\vtheta \in \mTheta^*(\kappa)} \left\lvert  \frac{\partial^3 R_t(\vtheta)}{\partial \vtheta_j \partial \vtheta_k \partial \vtheta_l} \right\rvert = O_p(1), \qquad j, k, l = 1,2,3,
\end{align}
for any fixed $\kappa \in (0, 1/2)$, see \citet[cor.\ 2.2]{New1991} and \citet[thm.\ 4.2]{Woo1994}.

The Hessian matrix can be derived from \eqref{eq:a3} and is given by
\begin{align}
\mH(\vtheta) &= \frac{\partial^2 R(\vtheta)}{\partial \vtheta \partial \vtheta'} =\frac{2}{n}\sum_{t=1}^{n}\left[
	\frac{\partial \varepsilon_t(\vtheta)}{\partial \vtheta}\frac{\partial \varepsilon_t(\vtheta)}{\partial \vtheta'} + 
	\varepsilon_t(\vtheta) \frac{\partial^2 \varepsilon_t(\vtheta)}{\partial \vtheta \partial \vtheta'}
\right],  \label{Hessian}
\end{align}
and a WLLN holds for the Hessian matrix if the absolute summability condition for \eqref{plim:fracdiff} is satisfied by the two different terms of the Hessian matrix. Since the coefficients of the first partial derivatives of $\varepsilon_{t}(\vtheta)$ were shown to be absolutely summable in the proof of theorem \ref{th:cons} for $\vtheta \in \mTheta^*(\kappa)$, the first term in \eqref{Hessian} directly satisfies the condition for  \eqref{plim:fracdiff} and thus is bounded in probability. It remains to be shown that absolute summability holds for the coefficients of $\partial^2 \varepsilon_t(\vtheta)/(\partial \vtheta \partial \vtheta')$. 
From \eqref{eq:a3}
\begin{align}\label{pd:eta2}
\begin{split}
\frac{\partial^2 \varepsilon_{t}(\vtheta)}{\partial \vtheta_j \partial \vtheta_k} =& 
\frac{\partial \psi_+(L, \vtheta)}{\partial \vtheta_j}\frac{\partial \Delta_+^{d-d_0}}{\partial \vtheta_k}\Delta_+^{d_0}y_t + 
\frac{\partial \psi_+(L, \vtheta)}{\partial \vtheta_k}\frac{\partial \Delta_+^{d-d_0}}{\partial \vtheta_j}\Delta_+^{d_0}y_t  \\
&+ 
\psi_+(L, \vtheta) \frac{\partial^2 \Delta_+^{d-d_0}}{\partial \vtheta_j \partial \vtheta_k}\Delta_+^{d_0}y_t + 
\frac{\partial^2 \psi_+(L, \vtheta)}{\partial \vtheta_j \partial \vtheta_k}\Delta_+^d y_t,
\end{split}
\end{align}
for $j, k =1,2,3$. The coefficients in $\partial \psi_+(L, \vtheta)/\partial \vtheta_j$ were already shown to be absolutely summable in the proof of theorem \ref{th:cons}, and thus the first and second term in \eqref{pd:eta2} satisfy the absolute summability condition for \eqref{plim:fracdiff}. As the coefficients in $\psi(L, \vtheta)$ are absolutely summable, the third term in \eqref{pd:eta2} is also bounded by \eqref{plim:fracdiff}, so that only the coefficients of the second partial derivatives of $\psi_+(L, \vtheta)$ need to be shown to be absolutely summable. As their sum is bounded by the sum of absolute coefficients of the untruncated polynomial $\psi(L, \vtheta) = \phi(L, \vtheta)^{-1}$, it is sufficient to prove absolute summability for the latter. For this, consider
\begin{align}\label{psi1}
	\frac{\partial^2 \psi(L, \vtheta)}{\partial \vtheta_j \partial \vtheta_k} = 2 \psi(L, \vtheta)^3 \frac{\partial \phi(L, \vtheta)}{\partial \vtheta_j}\frac{\partial \phi(L, \vtheta)}{\partial \vtheta_k} - \psi(L, \vtheta)^2 \frac{\partial^2 \phi(L, \vtheta)}{\partial \vtheta_j \partial \vtheta_k}, \qquad j, k = 1,2,3,
\end{align}
where the coefficients of first and second partial derivatives of $\phi(L, \vtheta)$ are shown to be absolutely summable in lemma \ref{LD} in appendix \ref{app:D}. Thus, \eqref{pd:eta2} satisfies the absolute summability condition for \eqref{plim:fracdiff}, so that the Hessian matrix \eqref{Hessian} satisfies a WLLN. 

To prove \eqref{Hessian:WLLN}, consider
\begin{align*}
	 \frac{\partial^3 R(\vtheta)}{\partial \vtheta_j \partial \vtheta_k \partial \vtheta_l} =\frac{2}{n}\sum_{t=1}^{n}\left[
	\frac{\partial^2 \varepsilon_t(\vtheta)}{\partial \vtheta_j \partial \vtheta_k}\frac{\partial \varepsilon_t(\vtheta)}{\partial \vtheta_l} + 
	\frac{\partial^2 \varepsilon_t(\vtheta)}{\partial \vtheta_j \partial \vtheta_l}\frac{\partial \varepsilon_t(\vtheta)}{\partial \vtheta_k} +
	\frac{\partial^2 \varepsilon_t(\vtheta)}{\partial \vtheta_k \partial \vtheta_l}\frac{\partial \varepsilon_t(\vtheta)}{\partial \vtheta_j} +
	\varepsilon_t(\vtheta) \frac{\partial^3 \varepsilon_t(\vtheta)}{\partial \vtheta_j \partial \vtheta_k \partial \vtheta_l} \right],
\end{align*}
$j, k, l = 1,2,3$, where absolute summability of the coefficients of the first three terms was already shown. Consequently, for the last term to also satisfy the condition for \eqref{plim:fracdiff}, the coefficients of the third partial derivatives of $\varepsilon_{t}(\vtheta)$ need to be shown to be absolutely summable. The derivatives are
\begin{align}\label{eq:eps3}
	\frac{\partial^3 \varepsilon_{t}(\vtheta)}{\partial \vtheta_j \partial \vtheta_k \partial \vtheta_l} = \frac{\partial^3 \psi_+(L, \vtheta)}{\partial \vtheta_j \partial \vtheta_k \partial \vtheta_l} \Delta_+^{d}y_t +
	\psi_+(L, \vtheta) \frac{\partial^3 \Delta_+^{d-d_0}}{\partial \vtheta_j \partial \vtheta_k \partial \vtheta_l}\Delta_+^{d_0}y_t + r_t(\vtheta),
\end{align}
and $r_t(\vtheta)$ holds the products of first and second partial derivatives of $\psi(L, \vtheta)$ and $\Delta_+^{d-d_0}$ that have already been shown to satisfy the absolute summability condition for \eqref{plim:fracdiff}. The second term in \eqref{eq:eps3} directly satisfies the condition for \eqref{plim:fracdiff}, so that only the first term remains to be checked. As before, the partial derivatives of the untruncated polynomial are considered, as they are an upper bound for the sum of absolute coefficients of the truncated polynomial. From \eqref{psi1}
\begin{align*}
		\frac{\partial^3 \psi(L, \vtheta)}{\partial \vtheta_j \partial \vtheta_k \partial \vtheta_l} =& 
		2\psi(L, \vtheta)^3 \left[\frac{\partial^2 \phi(L, \vtheta)}{\partial \vtheta_j \partial \vtheta_k}\frac{\partial \phi(L, \vtheta)}{\partial \vtheta_l}+
		\frac{\partial^2 \phi(L, \vtheta)}{\partial \vtheta_j \partial \vtheta_l}\frac{\partial \phi(L, \vtheta)}{\partial \vtheta_k}+
		 \frac{\partial^2 \phi(L, \vtheta)}{\partial \vtheta_k \partial \vtheta_l}\frac{\partial \phi(L, \vtheta)}{\partial \vtheta_j} \right] \\
		&-6 \psi(L, \vtheta)^4 \frac{\partial \phi(L, \vtheta)}{\partial \vtheta_j}\frac{\partial \phi(L, \vtheta)}{\partial \vtheta_k} \frac{\partial \phi(L, \vtheta)}{\partial \vtheta_l}
		- \psi(L, \vtheta)^2 \frac{\partial^3 \phi(L, \vtheta)}{\partial \vtheta_j \partial \vtheta_k \partial \vtheta_l}, \qquad j, k, l = 1,2,3.
\end{align*}
Absolute summability of the coefficients of the partial derivatives of $\phi(L, \vtheta)$ up to order three is shown in lemma \ref{LD} in appendix \ref{app:D}. Consequently, \eqref{eq:eps3} satisfies the absolute summability condition for \eqref{plim:fracdiff}, so that \eqref{Hessian:WLLN} holds. Thus, a UWLLN holds for the Hessian matrix, so that pointwise convergence generalizes to weak convergence. This, together with consistency of $\hat \vtheta$ (see theorem \ref{th:cons}) allows to evaluate the Hessian matrix in \eqref{eq:taylor} at $\vtheta_0$. Analogously to \eqref{eqn:Sn1}, it follows from the argument of \citet[pp.\ 135-136]{Rob2006} that the partial derivatives of $\varepsilon_{t}(\vtheta)$ in \eqref{Hessian} can be replaced by those of $\tilde \varepsilon_t(\vtheta)$ as $n \to \infty$, and $\varepsilon_{t}(\vtheta_0)$ can be replaced by $\varepsilon_t$, which yields
\begin{align}
\frac{\partial^2 R_t(\vtheta)}{\partial \vtheta_j \partial \vtheta_k}\Bigg\rvert_{\theta = \theta_0}&=  \frac{2}{n} \sum_{t=1}^n\left[
\frac{\partial \varepsilon_t(\vtheta)}{\partial \vtheta_j}\Bigg\rvert_{\theta = \theta_0}\frac{\partial \varepsilon_t(\vtheta)}{\partial \vtheta_k}\Bigg\rvert_{\theta = \theta_0} + 
\varepsilon_t(\vtheta_0) \frac{\partial^2 \varepsilon_t(\vtheta)}{\partial \vtheta_j \partial \vtheta_k}\Bigg\rvert_{\theta = \theta_0}
\right]  \pto 2 \mOmega^{(j, k)}_0, \label{plim:hessian}
\end{align}
as $n \to \infty$. The second term converges to zero in probability, as the second partial derivatives are $\tilde{\mathcal{F}}_{t-1}$-measurable, and thus the second term adapted to $\tilde{\mathcal{F}}_{t-1}$ is a stationary MDS.

Solving \eqref{eq:taylor} for $\sqrt{n}(\hat \vtheta - \vtheta_0)$ and plugging in the limits for first and second partial derivatives yields
\begin{align}
\sqrt{n}(\hat \vtheta - \vtheta_0) &= \mH_t(\bar \vtheta)
^{-1} \frac{1}{\sqrt{n}} \frac{\partial R(\vtheta)}{\partial \vtheta} \Bigg\rvert_{\vtheta = \vtheta_0} \dto  \mathrm{N}(0, \mOmega^{-1}_0),
\end{align}
as $n \to \infty$, which completes the proof.

\end{proof}
\section{Partial derivatives of $\phi(L, \vtheta)$}\label{app:D}
\begin{lemma}[Absolute summability of partial derivatives]\label{LD}
	For $\phi(L, \vtheta)$ in
	\begin{align}\label{eq:model}
		\phi(L, \vtheta)\sigma_\varepsilon\varepsilon_t^* = \sigma_\varepsilon\sum_{i=0}^{t-1}\phi_i(\vtheta) \varepsilon_{t-i}^* = \sigma_\eta \eta_t^* + \Delta_+^d \sigma_u u_t^* = \sigma_\eta \eta_t^* + \sigma_u\sum_{i=0}^{t-1} \pi_i(d)u_{t-i}^*,
	\end{align}
	with $\varepsilon_t^* \sim \mathrm{WN}(0, 1)$, $u_t^* \sim \mathrm{WN}(0, 1)$, $\eta_t^* \sim \mathrm{WN}(0, 1)$, $\phi_0(\vtheta) = 1$, it holds that 
	\begin{align}
		\lim_{t \to \infty}\sum_{i=1}^{t-1}\left\lvert \frac{\partial \phi_i(\vtheta)}{\partial \vtheta_j}\right\rvert &< \infty, \label{dphi1} \\
			\lim_{t \to \infty}\sum_{i=1}^{t-1}\left\lvert \frac{\partial^2 \phi_i(\vtheta)}{\partial \vtheta_j \partial \vtheta_k}\right\rvert &< \infty, \label{dphi2} \\
			\lim_{t \to \infty}\sum_{i=1}^{t-1}\left\lvert \frac{\partial^3 \phi_i(\vtheta)}{\partial \vtheta_j \partial \vtheta_k \partial \vtheta_l}\right\rvert &< \infty,  \label{dphi3}
	\end{align}
	for all $j, k, l = 1, 2, 3$, and all $\vtheta \in \mTheta$, where $\vtheta_j$ denotes the $j$-th entry of $\vtheta = (d, \sigma_\eta^2, \sigma_u^2)'$.
\end{lemma}

\begin{proof}[Proof of lemma \ref{LD}]
	The following results are required to prove \eqref{dphi1} to \eqref{dphi3}.
	For $\sigma_\varepsilon^2$, note that by solving the variance of \eqref{eq:model} for $\sigma_\varepsilon^2$
	\begin{align}\label{eq:sigma}
		\sigma_\varepsilon^2 = \frac{\sigma_\eta^2 + \sigma_u^2\sum_{i=0}^{t-1}\pi_i^2(d)}{\sum_{i=0}^{t-1}\phi_i^2(\vtheta)}.
	\end{align}
Since $\partial^j\pi_i(d)/\partial d^j = O(i^{-d-1}(1 + \log i)^j)$ for all $i \geq 1$, $j \geq 0$, see \citet[lemma B.3]{JohNie2010}, and thus $\lim_{t \to \infty}\sum_{i =1}^{t-1}\lvert\partial^j \pi_i(d)/\partial d^j\rvert < \infty$ for all $j \geq 0$, it follows that
\begin{align}
	\frac{\partial}{\partial \vtheta_j}\left[\sigma_\eta^2 +\sigma_u^2 \sum_{i=0}^{t-1}\pi_i^2(d)\right] &= O(1), \label{dsig1}\\
	\frac{\partial^2}{\partial \vtheta_j \partial \vtheta_k}\left[\sigma_\eta^2 + \sigma_u^2\sum_{i=0}^{t-1}\pi_i^2(d)\right] &= O(1), \label{dsig2}\\
	\frac{\partial^3}{\partial \vtheta_j \partial \vtheta_k \partial \vtheta_l}\left[\sigma_\eta^2 + \sigma_u^2\sum_{i=0}^{t-1}\pi_i^2(d)\right] &= O(1),\label{dsig3}
\end{align}
for all $j, k, l = 1,2,3$.
For the same reason, it follows from \eqref{eq:model} that 
\begin{align}\label{dphi4}
	\frac{\partial \phi(L, \vtheta)\sigma_\varepsilon\varepsilon_t^*}{\partial \vtheta_j} &= O(1)\varepsilon_t^* + \sum_{i=1}^{t-1}O(i^{-d-1}(1+\log i))\varepsilon_{t-i}^*,\\
	\frac{\partial^2 \phi(L, \vtheta)\sigma_\varepsilon\varepsilon_t^*}{\partial \vtheta_j \partial \vtheta_k} &= O(1)\varepsilon_t^* + \sum_{i=1}^{t-1}O(i^{-d-1}(1+\log i)^2)\varepsilon_{t-i}^*,\label{dphi5} \\
	\frac{\partial^3 \phi(L, \vtheta)\sigma_\varepsilon\varepsilon_t^*}{\partial \vtheta_j \partial \vtheta_k \partial \vtheta_l} &= O(1)\varepsilon_t^* + \sum_{i=1}^{t-1}O(i^{-d-1}(1+\log i)^3)\varepsilon_{t-i}^*,\label{dphi6}
\end{align}
and the limits stem from the first, second and third partial derivatives of $ \sigma_\eta \eta_t^* + \sigma_u\sum_{i=0}^{t-1} \pi_i(d)u_{t-i}^*$ w.r.t.\ $d$, while all coefficients of the other partial derivatives are bounded below. Consequently, \eqref{dphi4} to \eqref{dphi6} are MA processes with absolutely summable coefficients. Note that this is not sufficient for absolute summability of the partial derivatives of $\phi(L, \vtheta)$, as $\sigma_\varepsilon$ in the numerators of \eqref{dphi4} to \eqref{dphi6} also depends on $\vtheta$.

For \eqref{dphi1}, consider $\partial \sigma_\varepsilon^2/\partial \theta_j = c_{1}(\vtheta, \vtheta_j) - c_{2}(\vtheta, \vtheta_j)$, where
\begin{align}\label{eq:c}
c_{1}(\vtheta, \vtheta_j) = \frac{\frac{\partial}{\partial \vtheta_j} \left[ \sigma_\eta^2 + \sigma_u^2\sum_{i=0}^{t-1}\pi_i^2(d) \right]}{\sum_{i=0}^{t-1} \phi_i^2(\vtheta)} = O(1), \qquad 
c_{2}(\vtheta, \vtheta_j) = \frac{2\sigma_\varepsilon^2 \sum_{i=1}^{t-1}\phi_i(\vtheta)\frac{\partial \phi_i(\vtheta)}{\partial \vtheta_j}}{\sum_{i=0}^{t-1} \phi_i^2(\vtheta)},
\end{align}
and the first term is $O(1)$ due to \eqref{dsig1}. For the partial derivative of $\phi(L, \vtheta)\sigma_\varepsilon\varepsilon_t^*$ one then has
\begin{align}
	\frac{\partial \phi(L, \vtheta)\sigma_\varepsilon \varepsilon_{t}^*}{\partial \vtheta_j} &= \sigma_\varepsilon \sum_{i=1}^{t-1}\frac{\partial \phi_i(\vtheta)}{\partial \vtheta_j} \varepsilon_{t-i}^* + \frac{1}{2\sigma_\varepsilon} \frac{\partial \sigma_\varepsilon^2}{\partial \vtheta_j}\phi(L, \vtheta)\varepsilon_{t}^* \nonumber\\
	&= \frac{c_1(\vtheta, \vtheta_j)}{2\sigma_\varepsilon} \phi(L, \vtheta)\varepsilon_{t}^* - \frac{c_2(\vtheta, \vtheta_j)}{2\sigma_\varepsilon} \phi(L, \vtheta)\varepsilon_{t}^* + 
	\sigma_\varepsilon \sum_{i=1}^{t-1}\frac{\partial \phi_i(\vtheta)}{\partial \vtheta_j} \varepsilon_{t-i}^*.\label{pdphi}
\end{align}
From \eqref{dphi4} it follows that the term on the left hand side (LHS) is a MA process with absolutely summable coefficients for any $t$. Since the same holds for $\phi(L, \vtheta) \varepsilon_{t}^*$, by \eqref{eq:c} the first term on the right hand side (RHS) is also a MA process with absolutely summable coefficients. Consequently, the difference of the latter two terms on the RHS
\begin{align*}
	\sigma_\varepsilon \sum_{i=1}^{t-1}\frac{\partial \phi_i(\vtheta)}{\partial \vtheta_j} \varepsilon_{t-i}^* - 	\frac{c_2(\vtheta, \vtheta_j)}{2\sigma_\varepsilon} \phi(L, \vtheta)\varepsilon_{t}^*=
	\frac{-c_2(\vtheta, \vtheta_j)}{2\sigma_\varepsilon} \varepsilon_{t}^*  +
	\sum_{i=1}^{t-1}\left[  \sigma_\varepsilon \frac{\partial \phi_i(\vtheta)}{\partial \vtheta_j} - \frac{c_2(\vtheta, \vtheta_j)\phi_i(\vtheta)}{2\sigma_\varepsilon}  \right]\varepsilon_{t-i}^*,
\end{align*}
is also a MA process with absolutely summable coefficients. As the contemporaneous impact of $\varepsilon_{t}^*$ cannot cancel, it follows that $c_2(\vtheta, \vtheta_j)=O(1)$ is bounded, and thus the second term on the RHS of \eqref{pdphi} is a MA process with absolutely summable coefficients. For the equality in \eqref{pdphi} to hold, it must thus hold that $\sigma_\varepsilon \sum_{i=1}^{t-1}({\partial \phi_i(\vtheta)}/{\partial \vtheta_j} )\varepsilon_{t-i}^*$ is also a MA process with absolutely summable coefficients for any $t$, which proves \eqref{dphi1}. 

For \eqref{dphi2} one has $\partial^2 \sigma_\varepsilon^2/(\partial \vtheta_j \partial \vtheta_k) = c_3(\vtheta, \vtheta_j, \vtheta_k) - c_4(\vtheta, \vtheta_j, \vtheta_k) $ with
\begin{align}\begin{split}\label{eq:c2}
		c_3(\vtheta, \vtheta_j, \vtheta_k) =& \Bigg\{\frac{\partial^2}{\partial \vtheta_j \partial \vtheta_k} \left[\sigma_\eta^2 + \sigma_u^2 \sum_{i=0}^{t-1}\pi_i^2(d)\right]-
		2 \frac{\partial \sigma_\varepsilon^2}{\partial \vtheta_k} \sum_{i=1}^{t-1}\phi_i(\vtheta) \frac{\partial \phi_i(\vtheta)}{\partial \vtheta_j}\Bigg\}\left[\sum_{i=0}^{t-1}\phi_i^2(\vtheta)\right]^{-1}  \\
		&- 
		\Bigg\{2 \frac{\partial \sigma_\varepsilon^2}{\partial \vtheta_j} \sum_{i=1}^{t-1}\phi_i(\vtheta) \frac{\partial \phi_i(\vtheta)}{\partial \vtheta_k}  
		+2 \sigma_\varepsilon^2 \sum_{i=1}^{t-1} \frac{\partial \phi_i(\vtheta)}{\partial \vtheta_j} \frac{\partial \phi_i(\vtheta)}{\partial \vtheta_k}\Bigg\} \left[\sum_{i=0}^{t-1}\phi_i^2(\vtheta)\right]^{-1}  ,
	\end{split}\\
	c_4(\vtheta, \vtheta_j, \vtheta_k) =& \left[{2 \sigma_\varepsilon^2 \sum_{i=1}^{t-1} \phi_i(\vtheta) \frac{\partial^2 \phi_i(\vtheta)}{\partial \vtheta_j \partial \vtheta_k}}\right]\left[\sum_{i=0}^{t-1}\phi_i^2(\vtheta)\right]^{-1},
\end{align}
and $c_3(\vtheta, \vtheta_j, \vtheta_k) = O(1)$ is bounded due to \eqref{dphi1} and \eqref{dsig2}. The second partial derivatives of $\phi(L, \vtheta)\sigma_\varepsilon \varepsilon^*_t$ are 
\begin{align}\label{deriv2}
	\frac{\partial^2 \phi(L, \vtheta)\sigma_\varepsilon \varepsilon_{t}^*}{\partial \vtheta_j \partial \vtheta_k} =&
	z_{1}(\vtheta, \vtheta_j, \vtheta_k) + \frac{1}{2\sigma_\varepsilon} \frac{\partial^2 \sigma_\varepsilon^2}{\partial \vtheta_j \partial \vtheta_k}\phi(L, \vtheta) \varepsilon_{t}^* + \sigma_\varepsilon \sum_{i=1}^{t-1} \frac{\partial^2 \phi_i(\vtheta)}{\partial \vtheta_j \partial \vtheta_k} \varepsilon_{t-i}^*, \\
	z_{1}(\vtheta, \vtheta_j, \vtheta_k) =& \frac{1}{2\sigma_\varepsilon}\frac{\partial \sigma_\varepsilon^2}{\partial \vtheta_k} \sum_{i=1}^{t-1}\frac{\partial \phi_i(\vtheta)}{\partial \vtheta_j} \varepsilon_{t-i}^* + \frac{1}{2\sigma_\varepsilon}\frac{\partial \sigma_\varepsilon^2}{\partial \vtheta_j} \sum_{i=1}^{t-1}\frac{\partial \phi_i(\vtheta)}{\partial \vtheta_k} \varepsilon_{t-i}^* - \frac{1}{4\sigma_\varepsilon^3} \frac{\partial \sigma_\varepsilon^2}{\partial \vtheta_j}\frac{\partial \sigma_\varepsilon^2}{\partial \vtheta_k} \phi(L, \vtheta) \varepsilon_{t}^*, \nonumber
\end{align}
and $z_{1}(\vtheta, \vtheta_j, \vtheta_k)$ is a MA process with absolutely summable coefficients due to \eqref{dphi1}. Plugging in $\partial^2 \sigma_\varepsilon^2/(\partial \vtheta_j \partial \vtheta_k) = c_3(\vtheta, \vtheta_j, \vtheta_k) - c_4(\vtheta, \vtheta_j, \vtheta_k) $ and rearranging terms yields
\begin{align*}
	\frac{\partial^2 \phi(L, \vtheta)\sigma_\varepsilon \varepsilon_{t}^*}{\partial \vtheta_j \partial \vtheta_k} - \frac{c_3(\vtheta, \vtheta_j, \vtheta_k)}{2 \sigma_\varepsilon}\phi(L, \vtheta)\varepsilon_{t}^* &- z_1(\vtheta, \vtheta_j, \vtheta_k)= -\frac{c_4(\vtheta, \vtheta_j, \vtheta_k)}{2 \sigma_\varepsilon}\varepsilon_t^* \\
	&+ \sum_{i=1}^{t-1}\left[\sigma_\varepsilon \frac{\partial^2 \phi_i(\vtheta)}{\partial \vtheta_j \partial \vtheta_k} - \frac{c_4(\vtheta, \vtheta_j, \vtheta_k)}{2 \sigma_\varepsilon} \phi_i(\vtheta)\right] \varepsilon_{t-i}^*,
\end{align*}
where the LHS is a MA process with absolutely summable coefficients for any $t$ due to \eqref{dphi5} and \eqref{eq:c2}. Again, as the contemporaneous $\varepsilon_{t}^*$ cannot cancel out, $c_4(\vtheta, \vtheta_j, \vtheta_k) = O(1)$ is bounded. Therefore,  $c_4(\vtheta, \vtheta_j, \vtheta_k)/(2\sigma_\varepsilon) \phi(L, \vtheta)\varepsilon_t^*$ is a MA process with absolutely summable weights, so that for the equality above to hold, $\sum_{i =1}^{t-1}\partial^2 \phi_i(\vtheta)/(\partial \vtheta_j \partial \vtheta_k)\varepsilon_{t-i}^*$ must also be a MA process with absolutely summable weights for any $t$, which proves \eqref{dphi2}.

Turning to \eqref{dphi3}, the third partial derivatives of the variance parameter $\sigma_\varepsilon^2$ 
can be represented as  $\partial^3 \sigma_\varepsilon^2/(\partial \vtheta_j \partial \vtheta_k \partial \vtheta_l) = c_5(\vtheta, \vtheta_j, \vtheta_k, \vtheta_l) - c_6(\vtheta, \vtheta_j, \vtheta_k, \vtheta_l) $ with
\begin{align}
	c_6(\vtheta, \vtheta_j, \vtheta_k, \vtheta_l) =& \left[{2 \sigma_\varepsilon^2 \sum_{i=1}^{t-1} \phi_i(\vtheta) \frac{\partial^3 \phi_i(\vtheta)}{\partial \vtheta_j \partial \vtheta_k \partial \vtheta_l}}\right]\left[\sum_{i=0}^{t-1}\phi_i^2(\vtheta)\right]^{-1}.
\end{align}
$c_5(\vtheta, \vtheta_j, \vtheta_k, \vtheta_l)$ holds the products of first and second partial derivatives of $\sigma_\varepsilon^2$ and $\phi(1, \vtheta)$ that have already been shown to be $O(1)$, as well as ${\partial^3}/({\partial \vtheta_j \partial \vtheta_k \partial \vtheta_l}) \left[\sigma_\eta^2 + \sigma_u^2 \sum_{i=0}^{t-1}\pi_i^2(d)\right]$ that is $O(1)$ as shown in \eqref{dsig3}. Consequently $c_5(\vtheta, \vtheta_j, \vtheta_k, \vtheta_l) = O(1)$, and the exact expression is omitted for brevity. The third partial derivatives of $\phi(L, \vtheta) \sigma_\varepsilon \varepsilon_t^*$ follow from \eqref{deriv2} and equal
\begin{align}
	\frac{\partial^3 \phi(L)\sigma_\varepsilon \varepsilon_t^*}{\partial \vtheta_j \partial \vtheta_k \partial \vtheta_l} &= z_2(\vtheta, \vtheta_j, \vtheta_k, \vtheta_l) + \frac{1}{2\sigma_\varepsilon} \frac{\partial^3 \sigma_\varepsilon^2}{\partial \vtheta_j \partial \vtheta_k \partial \vtheta_l}\phi(L, \vtheta) \varepsilon_{t}^* + \sigma_\varepsilon \sum_{i=1}^{t-1} \frac{\partial^3 \phi_i(\vtheta)}{\partial \vtheta_j \partial \vtheta_k \partial \vtheta_l} \varepsilon_{t-i}^*,
\end{align}
where $z_2(\vtheta, \vtheta_j, \vtheta_k, \vtheta_l)$ holds the products of the first and second partial derivatives of $\sigma_\varepsilon^2$ and $\phi(L, \vtheta)$ for which absolute summability was shown above. Therefore, $z_2(\vtheta, \vtheta_j, \vtheta_k, \vtheta_l)$ is a MA process with absolutely summable coefficients. Plugging in $\partial^3 \sigma_\varepsilon^2/(\partial \vtheta_j \partial \vtheta_k \partial \vtheta_l) = c_5(\vtheta, \vtheta_j, \vtheta_k, \vtheta_l) - c_6(\vtheta, \vtheta_j, \vtheta_k, \vtheta_l) $ and rearranging gives 
\begin{align*}
	\frac{\partial^3 \phi(L, \vtheta)\sigma_\varepsilon \varepsilon_{t}^*}{\partial \vtheta_j \partial \vtheta_k \partial \vtheta_l} - \frac{c_5(\vtheta, \vtheta_j, \vtheta_k, \vtheta_l)}{2 \sigma_\varepsilon}\phi(L, \vtheta)\varepsilon_{t}^* &- z_2(\vtheta, \vtheta_j, \vtheta_k, \vtheta_l)= -\frac{c_6(\vtheta, \vtheta_j, \vtheta_k, \vtheta_l)}{2 \sigma_\varepsilon}\varepsilon_t^* \\
	&+ \sum_{i=1}^{t-1}\left[\sigma_\varepsilon \frac{\partial^3 \phi_i(\vtheta)}{\partial \vtheta_j \partial \vtheta_k \partial \vtheta_l} - \frac{c_6(\vtheta, \vtheta_j, \vtheta_k, \vtheta_l)}{2 \sigma_\varepsilon} \phi_i(\vtheta)\right] \varepsilon_{t-i}^*,
\end{align*}
where the LHS is a MA process with absolutely summable coefficients for any $t$ by \eqref{dphi6}. As for the first and second partial derivatives, $c_6(\vtheta, \vtheta_j, \vtheta_k, \vtheta_l) = O(1)$ holds, as the contemporaneous $\varepsilon_t^*$ do not cancel on the RHS. Due to boundedness of $c_6(\vtheta, \vtheta_j, \vtheta_k, \vtheta_l)$, the term $c_6(\vtheta, \vtheta_j, \vtheta_k, \vtheta_l) = O(1)\phi(L, \vtheta) \varepsilon_t^*$ is a MA process with absolutely summable weights. Since all other terms are MA processes with absolutely summable weights,  $\sum_{i=1}^{t-1} {\partial^3 \phi_i(\vtheta)}/({\partial \vtheta_j \partial \vtheta_k \partial \vtheta_l}) \varepsilon_{t-i}^*$ must also be a MA process with absolutely summable coefficients for the above equality to hold. This proves \eqref{dphi3}.


\end{proof}
\clearpage
\begin{spacing}{1}
\bibliographystyle{dcu}
\bibliography{literatur.bib}
\end{spacing}
\end{document}